\documentclass[12pt]{birkjour}
\usepackage{graphicx}
\usepackage{amssymb}
\usepackage{epstopdf,comment}
\usepackage{amsthm,  mathrsfs, enumerate, cite}

\usepackage{amsmath,color}
\usepackage{MnSymbol}

\newtheorem{definition}{Definition}
\newtheorem{proposition}{Proposition}
\newtheorem{theorem}{Theorem}
\newtheorem{corollary}{Corollary}
\newtheorem{lemma}{Lemma}
\newtheorem{assumption}{Assumption}
\newtheorem{remark}{Remark}

\newcommand{\req}[1]{Eq.\,(\ref{#1})}
\newcommand{\reqr}[2]{Eq.\,(\ref{#1}-\ref{#2})}

\numberwithin{theorem}{section}
\numberwithin{lemma}{section}
\numberwithin{equation}{section}
\numberwithin{proposition}{section}
\numberwithin{corollary}{section}

\begin{document} 
\title{Small Mass Limit of a Langevin Equation on a Manifold}

\author[Birrell]{Jeremiah Birrell}
\address{Department of Mathematics\\
University of Arizona\\
Tucson, AZ, 85721, USA}
\email{jbirrell@math.arizona.edu}

\author[Hottovy]{Scott Hottovy}
\address{Department of Mathematics \\
 University of Wisconsin-Madison \\
 Madison, WI 53706, USA}
\email{shottovy@math.wisc.edu}

\author[Volpe]{Giovanni Volpe}
\address{Soft Matter Lab, Department of Physics, Bilkent University, Ankara 06800, Turkey}
\address{UNAM -- National Nanotechnology Research Center, Bilkent University, Ankara 06800, Turkey}
\email{giovanni.volpe@fen.bilkent.edu.tr}

\author[Wehr]{Jan Wehr}
\address{Department of Mathematics\\
University of Arizona\\
Tucson, AZ, 85721, USA}
\email{wehr@math.arizona.edu}



\date{\today}

\begin{abstract}
We study damped geodesic motion of a particle of mass $m$ on a Riemannian manifold, in the presence of an external force and noise.  Lifting the resulting stochastic differential equation to the orthogonal frame bundle, we prove that, as $m \to 0$, its solutions converge to solutions of a limiting equation which includes a {\it noise-induced drift} term.  A very special case of the main result presents Brownian motion on the manifold as a limit of inertial systems.\end{abstract}
\maketitle

\section{Introduction}

Brownian motion (BM) plays a central role in many phenomena of scientific and technological significance. It lies at the foundation of stochastic calculus \cite{karatzas2014brownian}, which is applied to model a variety of phenomena, ranging from non-equilibrium statistical mechanics to  stock market fluctuations to population dynamics. In particular, Brownian motion occurs naturally in systems where microscopic and nanoscopic particles are present, as a consequence of thermal agitation \cite{Nelson1967}.  Brownian motion of micro- and nanoparticles occurring in complex environments  can often be represented as two-dimensional or one-dimensional manifolds embedded within a three-dimensional space. For example, the motion of proteins on cellular membranes occurs effectively on two-dimensional manifolds and is currently at the center of an intense experimental activity \cite{casuso2012characterization}. The single file diffusion of particles in  porous nanomaterials occurs in an effectively one-dimensional environment and plays a crucial role in many phenomena such as drug delivery, chemical catalysis and oil recovery \cite{k2012diffusion}. Several interesting phenomena can emerge in these conditions, such as  anomalous diffusion \cite{barkai2012single}, and  inhomogenous diffusion \cite{PhysRevX.5.011021}.  Similar phenomena also occur when considering active matter systems, such as living matter, which are characterised by being in a far-from-equilibrium steady state \cite{ramaswamy2010mechanics}; such systems often interact with complex environments that can be effectively modelled with low-dimensional manifolds embedded within a three-dimensional space. In order to gain a deeper understanding of these phenomena it is necessary to explore the properties of Brownian motion on manifolds.

The original motivation for this paper is to present Brownian motion on a manifold as the zero-mass limit of an inertial system.  Our main result is significantly more general and contains a rigorous version of the above statement as a special case.  In this section we will first outline this motivating problem and then discuss some earlier work on similar questions.  For the sake of clarity, we do not spell out all technical assumptions  and the arguments presented the introduction are heuristic.  

Brownian motion on an $n$-dimensional Riemannian manifold $(M,g)$ can be introduced as a mathematical object---a Markov process $x_t$ on $M$ with  the generator expressed in terms of the Riemannian metric $g$, which uniquely determines its law.  In this form, it has been a subject of an immense amount of study, both for its own sake and beauty, and for applications to analysis and geometry.  The reader is referred to  \cite{hsu2002stochastic, stroock2005introduction} and references therein.  In local coordinates, the components of BM satisfy the stochastic differential equation (SDE):
\begin{align}
dx_t^i = -{1 \over 2}g^{jk}\Gamma^i_{jk}\,dt + \sum_{\alpha=1}^n\sigma^i_{\alpha}\,dW^{\alpha}_t,
\end{align}
where $\sigma$ is the positive-definite square root of the inverse metric tensor $g^{-1}$, in the sense that $\sum_{\alpha = 1}^n\sigma^i_{\alpha} \sigma^k_{\alpha} = g^{ik}$, see p. 87 of \cite{hsu2002stochastic}. From the applied point of view, BM on a manifold is an idealized probabilistic description of diffusive motion performed by a particle constrained to $M$.  This can be justified at various levels, depending on what one is willing to assume.  Let us mention in particular the work of van Kampen \cite{vanKampen} which studies the conditions on the constraints, restricting motion of a diffusing particle to a manifold, under which its effective motion becomes Brownian. 

Here, as our point of departure, we take equations describing inertial motion of a particle of mass $m$, in the presence of two forces:  damping and noise.  The equations of motion in local coordinates are
\begin{align}
dx_t^i &= v_t^i\,dt,\label{SDE_motiv1} \\
m\,dv_t^i &= -m\Gamma^i_{jk}v_t^jv_t^k\,dt - \gamma^i_jv_t^j\,dt + \sum_{\alpha=1}^n\sigma_{\alpha}^i\,dW_t^{\alpha},\label{SDE_motiv2}
\end{align}
where $\Gamma^i_{jk}$ are the Christoffel's symbols of the metric $g$, $\gamma$ denotes the damping tensor and the vector fields $\sigma_\alpha$, $\alpha = 1, \dots  n$, couple the particle to $n$ standard Wiener processes, acting as noise sources.  The summation convention is used here and throughout the paper.  The reason the sum over $\alpha$ here is written explicitly is that it does not play  the role of a covariant index.

For the purposes of this motivating discussion, we assume that the damping and noise satisfy a fluctuation-dissipation relation known from nonequilibrium statistical mechanics \cite{kubo2012statistical}; this assumption will not be needed in the more general theorem that will be presented later. Note that the covariance of the noise is equal to $\sum_{\alpha}\sigma^i_{\alpha}\sigma^k_{\alpha}$ and is thus a tensor of type ${2 \choose 0}$ (a contravariant tensor of rank two).  We want to relate it to a quantity of the same type.  Since the damping tensor $\gamma^i_j$ has type ${1 \choose 1}$, we raise its lower index using the metric and state the fluctuation-dissipation relation as:
\begin{align}
\sum_{\alpha}\sigma^i_{\alpha}\sigma^j_{\alpha} = 2\beta^{-1} g^{jk}\gamma^i_k,
\end{align}
where $\beta^{-1} = k_BT$, $k_B$ is the Boltzmann constant and $T$ denotes the temperature.  In particular, if damping is isotropic,  $\gamma^i_k = \gamma\delta^i_k$, we get
\begin{align}\label{fluc_dis}
\sum_{\alpha}\sigma^i_{\alpha}\sigma^j_{\alpha} = 2\beta^{-1}\gamma g^{ij}
\end{align}
and \req{SDE_motiv2} becomes
\begin{align}\label{SDE_fluc_dis}
m\,dv^i_t = -m\Gamma^i_{jk}v^j_tv^k_t\,dt - \gamma v_t^i\,dt + \sum_{\alpha}\sigma^i_{\alpha}\,dW_t^{\alpha},
\end{align}
with $\sigma^i_{\alpha}$ satisfying the relation \req{fluc_dis}.
The problem motivating this work can now be stated as follows.
Consider the solutions of \req{SDE_motiv1} and \req{SDE_fluc_dis} with the initial conditions $x^{(m)}(0) = x_0$, $v^{(m)}(0)=v_0$.  We want to show that, as $m \to 0$, $x^{(m)}$ converges to the solution of the SDE
\begin{align}
dx_t^i = -{1 \over \beta \gamma^2}g^{jk}\Gamma^i_{jk}\,dt + {1 \over \gamma}\sum_{\alpha}\sigma^i_{\alpha}\,dW^{\alpha}_t
\end{align}
with the same initial condition.  Since this equation describes (a rescaled) BM on the manifold $M$ in local coordinates, this will realize our original goal.  Related results in the physics content are reported in \cite{polettini2013generally}.

We now present a sketch of the argument.  The remarks that follow it explain its relation to the actual proof in later sections.  Our guiding principle is that the kinetic energy of the particle is of order $1$, so that the components of the velocity (in a fixed coordinate chart) behave as  ${1 \over \sqrt{m}}$ in the limit $m \to 0$.

 Solving for $v_t^i\,dt$ in  \req{SDE_fluc_dis}, we obtain
\begin{align}\label{sol_outline}
dx_t^i = v_t^i\,dt = -{m \over \gamma}\,dv^i_t - {m \over \gamma}\Gamma^i_{jk}v_t^jv_t^k\,dt + {1 \over \gamma}\sum_{\alpha}\sigma^i_{\alpha}\,dW^{\alpha}_t. 
\end{align}
In the limit $m \to 0$ we expect no contribution from the first term, since $\gamma$ is constant and thus ${m \over \gamma}\,dv_t^i$ is the differential of the expression ${1 \over \gamma}mv_t^i$ which vanishes in the limit.    We {\it do} expect a nonzero limit  from the quadratic term.  This is again based on the analogy with \cite{Hottovy2014}, where  such a term appears as a result of integration by parts.  In the case discussed here, it is present in the equation from the start, reflecting the manifold geometry.    In the limit, we expect the fast velocity variable to average, giving rise to an $x$-dependent drift term.  To calculate this term, we use the method of \cite{Hottovy2014}, together with the heuristics that $v^i$ should be of order $m^{-{1 \over 2}}$.  Consider the differential (vanishing in the $m \to 0$ limit):
\begin{align}
&d(mv_t^jmv_t^k) = d(mv_t^j)mv_t^k + mv_t^j\,d(mv_t^k) + d(mv_t^j)\,d(mv_t^k)\notag \\
=& \big(-m\Gamma^j_{li}v^l_tv^i_t\,dt - \gamma v_t^j\,dt + \sum_{\alpha}\sigma^j_{\alpha}\,dW_t^{\alpha}\big)mv_t^k\\
& + mv_t^j\big( -m\Gamma^k_{li}v^l_tv^i_t\,dt - \gamma v_t^k\,dt + \sum_{\alpha}\sigma^k_{\alpha}\,dW_t^{\alpha}\big) + \sum_{\alpha}\sigma^j_{\alpha}\sigma^k_{\alpha}\,dt. \notag
\end{align}
Based on our assumption about the order of magnitude of $v^i$, in the limit we get
\begin{equation}
-2m\gamma v_t^jv_t^k\,dt + \sum_{\alpha}\sigma^j_{\alpha}\sigma^k_{\alpha}\,dt = 0.
\end{equation}
Substituting this into \req{sol_outline} and leaving out the terms which vanish in the limit, we obtain
\begin{equation}\label{limit_eq}
dx_t^i = -{1 \over \beta \gamma^2}g^{jk}\Gamma^i_{jk}\,dt + {1 \over \gamma}\sum_{\alpha}\sigma^i_{\alpha}\,dW^{\alpha}_t,
\end{equation}
which describes a rescaled BM on the manifold.

The limiting process $x_t$ satisfies an equation driven by the same Wiener processes, $W^\alpha$, that drove the equations for the original processes $x^{(m)}$, and is thus defined on the same probability space.  The processes  $x^{(m)}$ will be proven to converge to $x$ in the sense that the $L^p$-norm (in the $\omega$ variable) of the uniform distance on $[0,T]$ between the realizations of $x$ and $x^{(m)}$ goes to zero for every $T$.  This is much stronger than convergence in law on compact time intervals.

Zero-mass limits of diffusive systems have been studied in numerous works, starting from \cite{smoluchowski1916drei}. See \cite{Nelson1967} for a masterly review of the early history.   The analysis of such models have been extended in many directions. For example, \cite{doi:10.1137/S1540345903421076} studies the limit of a particle system driven by a fluid model that is coupled to noise.  The extension of diffusion processes to the relativistic setting has been studied in, for example, \cite{Chevalier2008,bailleul2010stochastic}. Other related works study random perturbations of the geodesic flow on a Riemannian manifold. In \cite{pinsky1976isotropic},  convergence of the transition semigroups of a family of transport processes on a manifold to that of Brownian motion was shown. See also the paper \cite{pinsky1981homogenization} where homogenization of the velocity variable for equations on manifolds is studied.   Families of Orstein-Uhlenbeck processes on manifolds were studied in \cite{Jorgensen1978,dowell1980differentiable}.   Two interesting recent papers  are \cite{XueMei2014} and \cite{angst2015kinetic}.  They prove convergence in law to Brownian motion in appropriate limits. We remark that in the special case of constant damping (as considered in this introduction), the generator of the process defined by \req{SDE_motiv1}-\req{SDE_motiv2} is the {\it hypoelliptic Laplacian}, introduced  in \cite{bismut2005hypoelliptic}---an important analytical object, encoding geometric properties of the manifold.  In this case, \cite{bismut2005hypoelliptic} proves a convergence-in-law result.  This has been extended to cover the convergence of kernels and their derivatives  \cite{bismut2015}. 

 More recently, several authors address the case when a general position dependent forcing is included and the damping and/or noise coefficients depend on the state of the system.  In particular, they study the associated phenomenon of the {\it noise-induced drift} that arises in the limit.  See references in the recent paper \cite{Hottovy2014}, where a formula for the noise-induced drift has been established for a large class of systems in Euclidean space of an arbitrary dimension.  See also \cite{herzog2015small}, where some of the assumptions made in \cite{Hottovy2014} are relaxed. 

The analysis applied here is most similar to that of \cite{Hottovy2014}, with two important differences.   First, we prove the fundamental momentum (or:  kinetic energy) bound in a different way, which would also lead to an alternative proof of the main result of \cite{Hottovy2014}.  Secondly, in this paper we pose the problem on a compact Riemannian manifold. This leads to complications of a  geometric nature that are absent in Euclidean space:  in order to control the quadratic terms in the geodesic equation more efficiently, we lift the equations to the orthogonal frame bundle. Another, equally important, consequence of lifting to the orthogonal frame bundle is that the equations of motion, including the noise term, can be formulated geometrically, without reference to local coordinate charts.   While this increases the number of variables and makes the equations more complicated, it simplifies the analytical aspects of the problem.  Frame bundle techniques similar to this have been used by many other authors, see for example \cite{Jorgensen1978,dowell1980differentiable,hsu2002stochastic,bailleul2010stochastic,XueMei2014}. 

 As our main result, we derive the equation satisfied by an inertial system in the limit $m\rightarrow 0$. In the case of constant damping and noise (and thus satifying the fluctuation-dissipation relation) with zero forcing, we obtain Brownian motion, as suggested by the informal derivation leading to \req{limit_eq}.  Various formulations that cover this classical case have been studied by previous authors \cite{dowell1980differentiable,bismut2005hypoelliptic}.  Forcing terms were also considered in \cite{dowell1980differentiable}.  The theorem presented here is general enough to include the classical case, as well as cover position dependent forcing, damping, and noise, including a derivation of the formula for the  noise induced drift. In addition, the damping and the noise coefficients are not required to satisfy a fluctuation-dissipation relation (unlike in the motivating discussion above), which leads to a fully general formula for the resulting noise-induced drift in the small mass limit. Physically, this result is particularly relevant when considering active matter and systems far from thermodynamic equilibrium \cite{seifert2012stochastic}.

\subsection{Summary of the Main Result}\label{sec:result_summary}
In this section we summarize the main result of this paper, including the required assumptions.  Motivation for the equations under consideration, further detail on the notation, and a proof of the result will follow in subsequent sections.

Let $(M,g)$ be a compact, connected $n$-dimensional Riemannian manifold without boundary, $F_O(M)$ be its frame bundle with cannonical projection $\pi$ (see Section \ref{sec:frame_bundle}), and $N\equiv F_O(M)\times\mathbb{R}^n$.  For each $m>0$ (representing the particle mass) we will  consider the following SDE on $N$ for fixed, non-random initial condition $(u_0,v_0)\in N$:
\begin{align}
u^m_t=&u_0+\int_{t_0}^t H_{v^m_s}(u^m_s)ds,\label{THE_SDE1}\\
v^m_t=&v_0+\frac{1}{m}\int_{t_0}^t [F(u^m_s)-\gamma(u^m_s)v^m_s]ds+\frac{1}{m}\int_{t_0}^t \sigma(u^m_s) dW_s.\label{THE_SDE2}
\end{align}
Here $H_\alpha(u)$ are the canonical horizontal vector fields on $F_O(M)$ (see Lemma \ref{horizontal_vf}), $F(x)$ is a smooth vector field on $M$ (the forcing), $\gamma(x)$ is a smooth $\binom{1}{1}$  tensor field on $M$ (the damping), the noise coefficients are given by a $\mathbb{R}^{n\times k}$-valued function,  $\sigma(u)$,  on $F_O(M)$,  $W_t$ is a $\mathbb{R}^k$-valued Wiener process, and we define $F(u)=u^{-1}F(\pi(u))$ and $\gamma(u)=u^{-1}\gamma(\pi(u))u$. 

Let $(u_t^m,v_t^m)$ be a family of solutions to \req{THE_SDE1}-\req{THE_SDE2}, corresponding to mass values $m>0$. To study the $m\rightarrow 0$ limit of $u_t^m$, we assume that the symmetric part of $\gamma$, $\gamma^s=\frac{1}{2}(\gamma+\gamma^T)$, has eigenvalues bounded below by a positive constant.  This coercivity assumption is crucial to our results, as it provides the damping necessary to ensure that the momentum degrees of freedom, $p_t^m=mv_t^m$, become negligible in the limit (see Section \ref{sec:p_zero}). 

Under the assumptions stated above, the main result (Theorem \ref{Main_theorem}) gives the limiting behavior of $u^m_t$ as $m\rightarrow 0$.  More specifically, we prove the following:\\
Fix $T> 0$ and a Riemannian metric tensor field on $F_O(M)$. Let $d$ be the associated metric on the connected component of $F_O(M)$ that contains $u_0$.  Then for any $q>0$ and any $0<\kappa< q/2$ we have 
\begin{align}
E[\sup_{t\in[0,T]}d(u^m_t,u_t)^{q}]=O(m^\kappa)\text{ as } m\rightarrow 0,
\end{align}
where $u_t$ solve the following SDE on $F_O(M)$ with initial condition $u_0$:
\begin{align}\label{THE_limit_SDE}
du_t=&H_{(\gamma^{-1} F)(u_t)} (u_t)dt+S(u_t)dt+H_{(\gamma^{-1}\sigma)(u_t)} (u_t)\circ dW_t.
\end{align}
The additional drift vector field, $S(u)$, on $F_O(M)$ that arises in the limit, given by \req{S_h}-\req{S_v},  will be called the {\em noise induced drift}.

\section{Forced Geodesic Motion on the Tangent Bundle}
We now begin the task of making the results outlined in the introduction and summary precise. Let $(M,g)$ be an $n$-dimensional smooth connected Riemannian manifold with tangent bundle $(TM,\pi)$, where $\pi$ is the natural projection. Let $V:TM\rightarrow TM$ be smooth and  $\pi\circ V=\pi$, i.e. $V$ maps each fiber into itself.  The deterministic dynamical system that we eventually want to couple to noise  is defined by the equation
\begin{align}\label{ode_system}
\nabla_{\dot{x}}\dot{x}=V(\dot{x}),
\end{align}
where $\nabla$ is the Levi-Civita connection.  In this section and the next, we focus on the non-random system. The coupling to noise will be discussed in Section \ref{sec:SDE_manifold}. We refer to the system \req{ode_system} as a geodesic equation with (velocity-dependent) forcing  $V$. Note that $\dot{x}$ is an element of  $TM$, and so it contains both position and velocity information.   In particular, \req{ode_system} contains the special case where $V$ is independent of the velocity degrees of freedom, i.e. $V$ is a vector field on $M$.

 \req{ode_system} is more general than the system outlined in the introduction, where the deterministic forcing consists only of drag.  The process that we will eventually find in the small mass limit will therefore be more general than Brownian motion on $M$, but will include Brownian motion as a special case.

We now interpret  \req{ode_system} as an ordinary differential equation (ODE) on the tangent bundle $TM$.  With $V \equiv 0$ it is the standard geodesic equation.  Below we give some facts, starting from this case in points 1-4 and then, in points 5-6, we add forcing. These facts will not be used in our subsequent analysis, but they give one an idea of how  forced geodesic motion on  manifold, \req{ode_system}, can be reformulated as a flow on a larger space.  We will build on this idea in the next section.
\begin{enumerate}
\item  For $v\in TM$ let $x_v$ be the geodesic with velocity $v$ at $t=0$.  Define the geodesic vector field $G:TM\rightarrow T(TM)$ by $G(v)=\frac{d}{dt}(\dot x_v)_{t=0}$, i.e. the tangent vector to the curve $\dot x_v:I\rightarrow TM$ at $t=0$. $G$ is a smooth vector field on $TM$ and $x:I\rightarrow M$ is a geodesic iff $\dot x$ (interpreted as a curve in $TM$) is an integral curve of $G$. 
\item If $\eta$ is an integral curve of $G$ then $x\equiv\pi\circ\eta$ is a geodesic on $M$ and $\dot{x}=\eta$.
\item The flow of $G$ is $(t,v)\rightarrow \dot{x}_v(t)$.
\item In a chart $x^i$ for $M$ and induced coordinates $(x^i,v^i)$ on $TM$, $G$ takes the form
\begin{equation}
G(x,v)=v^i\partial_{x_i}|_{(x,v)}-\Gamma^i_{jk}(x)v^jv^k\partial_{v^i}|_{(x,v)},
\end{equation}
where $\Gamma^i_{jk}$ are the Christoffel symbols of the Levi-Civita connection in the coordinate system $x^i$.
\item Using $V$ we can define a vector field $V:TM\rightarrow T(TM)$ given in an induced chart on $TM$ by $V=V^i\partial_{v^i}$ (we will let context dictate whether we consider $V$ as mapping into $TM$ or $T(TM)$). This produces  a well defined smooth vector field on $TM$ that is independent of the choice of charts.

\item We let 
\begin{align}\label{Y_def}
Y=G+V.
\end{align} 
If $\tau$ is an integral curve of $Y$ then  $x=\pi\circ\tau$ satisfies \req{ode_system} and $\tau=\dot{x}$. Conversely, if $x$ satisfies \req{ode_system} then $\dot{x}$ is an integral curve of $Y$.
\end{enumerate}
This last point implies that the equation of interest, \req{ode_system}, defines a smooth dynamical system on the tangent bundle of $M$ and the vector field of this dynamical system is $Y$.

\section{Forced Geodesic Motion on the Frame Bundle}\label{sec:frame_bundle}
The metric tensor $g$ on $M$ defines a reduction of the structure group of $TM$ to the orthogonal group, $O(\mathbb{R}^n)$ \cite{kobayashi2009foundations}, with local trivializations induced by local orthonormal (o.n.) frames on $M$ (i.e. collections of local vector fields that form an o.n. basis at each point of their domain).  In turn, this lets one construct the orthogonal frame bundle,  $(F_O(M),\pi)$ (we will let context distinguish between the various projections $\pi$). By reformulating \req{ode_system} as a dynamical system using the orthogonal frame bundle  of $M$, in a similar manner to the procedure outlined in the previous section, we will arrive at equations that are more amenable to being coupled to  noise.

Our expanded system will be defined via a vector field on the manifold $N\equiv  F_O(M)\times \mathbb{R}^n$.  
\subsection{Coordinate Independent Definition}\label{sec:X_def1}
Fix $(u,v)\in N$.  We will define a vector $X_{(u,v)}\in T_{(u,v)}N$ as follows.  Let $x(t)$ be the solution to 
\begin{equation}
\nabla_{\dot{x}}\dot{x}=V(\dot{x}),\hspace{2mm} x(0)=\pi(u),\hspace{2mm}\dot{x}(0)=u(v),
\end{equation}
 i.e. the integral curve of $Y$, defined in \req{Y_def}, starting at $u(v)\in T_{\pi(u)}M$.

  Let $U_\alpha(t)$ be the parallel translates of $u(e_\alpha)$ along $x(t)$ ($e_\alpha$ is the standard basis for $\mathbb{R}^n$), i.e.
\begin{align}
\nabla_{\dot{x}}U_\alpha=0,\hspace{2mm} U_\alpha(0)=u(e_\alpha).
\end{align}
  Parallel transport via the Levi-Civita connection preserves inner products, so $\tau(t)$ defined by $\tau(t)e_\alpha=U_\alpha(t)$ is a smooth section of $F_O(M)$ along $x(t)$.  Define the smooth curve in $\mathbb{R}^n$, $v(t)=\tau(t)^{-1}\dot{x}(t)$.

With these definitions, $\eta(t)=(\tau(t),v(t))$ is a smooth curve in $N$ starting at $(u,v)$.  Define the vector field $X$ by
\begin{align}\label{X_def}
X_{(u,v)}=\dot{\eta}(0).
\end{align}

\subsection{Coordinate Expression}\label{sec:ode_in_coords}
We now derive a formula for $X$ in a coordinate system defined below and thereby prove it is a smooth vector field on $N$.

Let $(U,\phi)$ be a coordinate chart on $M$ and $E_\alpha$ be an o.n. frame on $U$. We will let Roman indices denote quantities in the coordinate frame and Greek indices denote quantities in the local o.n. frame. The connection coefficients in the o.n. frame, $A^\alpha_{\beta\eta}$, are defined by $\nabla_{E_\beta} E_\eta=A^{\alpha}_{\beta\eta}E_\alpha$. The coordinate frame, $\partial_i$, and the o.n. frame, $E_\alpha$, are related by an invertible matrix valued smooth function $\Lambda^\alpha_i$ on $U$,
\begin{align}
\partial_i=\Lambda^\alpha_iE_\alpha.
\end{align}
Let $\psi$ be the local section of $F_O(M)$ induced by $E_\alpha$, i.e.  $\psi(x)v=v^\alpha E_\alpha(x)$.  We have the diffeomorphism $\Phi:\pi^{-1}(U)\rightarrow\phi(U)\times O(\mathbb{R}^n)$, $u\rightarrow(\phi(\pi(u)),h)$, where $h$ is uniquely defined by $u=\psi(\pi(u))h$.  In turn, this gives a diffeomorphism $\Phi\times id $ on $\pi^{-1}(U)\times \mathbb{R}^n\subset N$.

\begin{lemma}\label{X_coords_lemma}
The pushforward of the vector field $X$ to $\phi(U)\times O(\mathbb{R}^n)\times \mathbb{R}^n$ by the diffeomorphism $\Phi\times id$ is given by
\begin{align}\label{X_coords}
((\Phi\times id)_*X)|_{(x,h,v)}=&(\Lambda^{-1})_\alpha^j(\pi(u))h_\beta^\alpha v^\beta\partial_j-h_\alpha^\eta A_{\delta\eta}^\beta(\pi(u))h_\xi^\delta v^\xi\partial_{e^\beta_\alpha}\\
&+(h^{-1})^\alpha_\beta V^\beta(u(v))\partial_{v^\alpha}\notag
\end{align}
where $u=\Phi^{-1}(\phi^{-1}(x),h)$, $V^\beta$ are the components of $V$ in the o.n. frame $E_\beta$ (not the coordinate frame $\partial_i$), $v^\alpha$ are the standard coordinates on $\mathbb{R}^n$, and $e^\beta_\alpha$ are the standard coordinates on $\mathbb{R}^{n\times n}$.  In particular, $X$ is a smooth vector field. 
\end{lemma}
\begin{proof}
Using 
\begin{align}\label{tau_clarif}
\tau(t)=\psi(x(t))h(t)=u(t),\hspace{2mm} v(t)=\tau^{-1}(t)\dot{x}(t)
\end{align}
 we obtain
\begin{align}\label{v_def}
v(t)=&h^{-1}(t) \psi^{-1}(x(t)) \dot{x}^i(t)\partial_{i}= \dot{x}^i(t)\Lambda_i^\alpha(x(t))h^{-1}(t)\psi^{-1}(x(t)) E_\alpha(x(t))\notag\\
=&\dot{x}^i(t)\Lambda_i^\alpha(x(t)) h^{-1}(t)e_\alpha=\dot{x}^i(t)\Lambda_i^\alpha(x(t))(h^{-1})_\alpha^\beta(t)e_\beta.
\end{align}
Solving for $\dot{x}(t)$ we find
\begin{align}\label{gamma_dot_eq}
\dot{x}^j(t)=(\Lambda^{-1})_\alpha^j(x(t))h_\beta^\alpha(t)v^\beta(t).
\end{align}
This proves that the first term of \req{X_coords} is correct.

$\nabla_{\dot{x}}U_\alpha=0$ implies
\begin{align}
0=&\nabla_{\dot{x}} \psi(x(t))h(t) e_\alpha=\nabla_{\dot{x}}h_\alpha^\beta(t) E_\beta(x(t))\\
=&\dot{h}_\alpha^\beta(t) E_\beta(x(t))+h_\alpha^\beta(t) \dot{x}^k(t)\nabla_{\partial_k} E_\beta\notag\\
=&\dot{h}_\alpha^\beta(t) E_\beta(x(t))+h_\alpha^\beta(t) \dot{x}^k(t)\Lambda_k^\eta(x(t))\nabla_{E_\eta} E_\beta\notag\\
=&\left(\dot{h}_\alpha^\beta(t)+h_\alpha^\eta(t) \dot{x}^k(t)\Lambda_k^\delta(x(t))A_{\delta\eta}^\beta(x(t))\right) E_\beta(x(t)).\notag
\end{align}
Therefore
\begin{align}\label{g_dot_eq}
\dot{h}_\alpha^\beta(t)=-h_\alpha^\eta(t) A_{\delta \eta}^\beta(x(t))\Lambda_k^\delta(x(t))\dot{x}^k(t).
\end{align}
Using \req{gamma_dot_eq}, we get
\begin{align}\label{g_dot_eq2}
\dot{h}_\alpha^\beta(t)=&-h_\alpha^\eta(t) A_{\delta \eta}^\beta(x(t))\Lambda_k^\delta(x(t))((\Lambda^{-1})_\kappa^k(x(t))h_\xi^\kappa(t)v^\xi(t))\\
=&-h_\alpha^\eta(t) A_{\delta \eta}^\beta(x(t)) h_\xi^\delta(t)v^\xi(t).\notag
\end{align}
This proves that the second term in \req{X_coords} is correct.

Differentiating \req{v_def} (and dropping the time dependence in our notation) we find
\begin{align}
\dot{v}^\alpha=-(h^{-1})^\alpha_\beta \dot{h}^\beta_\eta (h^{-1})^\eta_\xi\Lambda^\xi_i\dot{x}^i+(h^{-1})^\alpha_\beta\partial_l\Lambda^\beta_i \dot{x}^l\dot{x}^i+(h^{-1})^\alpha_\beta\Lambda^\beta_i\ddot{x}^i.
\end{align}
From $\nabla_{\dot{x}}\dot{x}=V$ we obtain $\ddot{x}^i+\Gamma^i_{jk}\dot{x}^j\dot{x}^k=V^i$ where $V^i$ are the components of $V$ in the coordinate frame $\partial_i$ (not to be confused with $V^\alpha$, the components in the o.n. frame $E_\alpha$).  We need to convert from $\Gamma^i_{jk}$ to $A^\alpha_{\beta\eta}$,
\begin{align}
\Gamma^i_{jk}\partial_i=&\nabla_{\partial_j}\partial_k=\nabla_{\partial_j}\Lambda_k^\alpha E_\alpha=\partial_j\Lambda^\alpha_k E_\alpha+\Lambda^\alpha_k\nabla_{\partial_j}E_\alpha=\partial_j\Lambda^\alpha_k E_\alpha+\Lambda^\alpha_k\Lambda_j^\beta\nabla_{E_\beta}E_\alpha\notag\\
=&(\partial_j\Lambda^\alpha_k +\Lambda^\eta_k\Lambda_j^\beta A_{\beta \eta}^\alpha )E_\alpha=(\partial_j\Lambda^\alpha_k +\Lambda^\eta_k\Lambda_j^\beta A_{\beta\eta}^\alpha )(\Lambda^{-1})_\alpha^i \partial_i.
\end{align}
Using this we obtain
\begin{align}
\dot{v}^\alpha=&-(h^{-1})^\alpha_\beta \dot{h}^\beta_\eta (h^{-1})^\eta_\delta\Lambda^\delta_i\dot{x}^i+(h^{-1})^\alpha_\beta\partial_l\Lambda^\beta_i \dot{x}^l\dot{x}^i+(h^{-1})^\alpha_\beta\Lambda^\beta_i(V^i-\Gamma^i_{jm}\dot{x}^j\dot{x}^m)\notag\\
=&-(h^{-1})^\alpha_\beta \dot{h}^\beta_\eta (h^{-1})^\eta_\delta\Lambda^\delta_i\dot{x}^i+(h^{-1})^\alpha_\beta\partial_l\Lambda^\beta_i \dot{x}^l\dot{x}^i\\
&+(h^{-1})^\alpha_\beta\Lambda^\beta_i(V^i-(\partial_j\Lambda^\eta_m +\Lambda^\delta_m\Lambda_j^\xi A_{\xi \delta}^\eta )(\Lambda^{-1})_\eta^i \dot{x}^j\dot{x}^m)\notag\\
=&(h^{-1})^\alpha_\beta\Lambda^\beta_iV^i-(h^{-1})^\alpha_\beta \dot{h}^\beta_\eta (h^{-1})^\eta_\delta\Lambda^\delta_i\dot{x}^i+(h^{-1})^\alpha_\beta\partial_l\Lambda^\beta_i \dot{x}^l\dot{x}^i\notag\\
&-(h^{-1})^\alpha_\beta(\partial_j\Lambda^\beta_m +\Lambda^\delta_m\Lambda_j^\xi A_{\xi \delta}^\beta ) \dot{x}^j\dot{x}^m\notag\\
=&(h^{-1})^\alpha_\beta\Lambda^\beta_iV^i-(h^{-1})^\alpha_\beta \dot{h}^\beta_\eta (h^{-1})^\eta_\xi\Lambda^\xi_i\dot{x}^i+(h^{-1})^\alpha_\beta\partial_l\Lambda^\beta_i \dot{x}^l\dot{x}^i\notag\\
&-(h^{-1})^\alpha_\beta\partial_j\Lambda^\beta_m\dot{x}^j\dot{x}^m -(h^{-1})^\alpha_\beta\Lambda^\eta_m\Lambda_j^\xi A_{\xi\eta}^\beta  \dot{x}^j\dot{x}^m.\notag
\end{align}
The third and fourth terms cancel.  Using \req{g_dot_eq}, the second can be written
\begin{align}
&-(h^{-1})^\alpha_\beta \dot{h}^\beta_\eta (h^{-1})^\eta_\xi\Lambda^\xi_j\dot{x}^j =(h^{-1})^\alpha_\beta (h_\xi^\eta A_{\delta\eta}^\beta\Lambda_i^\delta\dot{x}^i)(h^{-1})^\xi_\epsilon\Lambda^\epsilon_j\dot{x}^j=(h^{-1})^\alpha_\beta A_{\xi \eta}^\beta\Lambda_i^\xi\Lambda^\eta_j\dot{x}^i\dot{x}^j.
\end{align}
Therefore
\begin{align}
\dot{v}^\alpha(t)=&(h^{-1})^\alpha_\beta\Lambda^\beta_iV^i+(h^{-1})^\alpha_\beta A_{ \xi \eta}^\beta\Lambda_i^\eta\Lambda^\xi_j\dot{x}^i\dot{x}^j-(h^{-1})^\alpha_\beta A_{ \xi \eta}^\beta \Lambda_j^\eta\Lambda^\xi_m\dot{x}^j\dot{x}^m\notag\\
=&(h^{-1})^\alpha_\beta\Lambda^\beta_iV^i=(h^{-1})^\alpha_\beta V^\beta.
\end{align}
Note that we have converted from the components in the coordinate basis, $V^i$, to the coordinates in the o.n. basis, $V^\beta$,  in the last line.
This proves  that the final term of \req{X_coords} is correct.

Note that the equation for $v(t)$ can also be written written as
\begin{align}
\dot{v}(t)=\tau^{-1}(t) V(\tau(t)v(t)).
\end{align}

\end{proof}
The cancellation of the Christoffel terms in the equation for $\dot{v}$ is not unexpected.  In the absence of forcing $V$, $x(t)$ is a geodesic and hence its tangent vector is parallel transported along itself.  Therefore  $v^i$, the components of the tangent vector in the parallel transported frame, $U_\alpha$, must be constants when $V$ vanishes. This is in contrast to the geodesic equation in an arbitrary coordinate system, in which the equation for $\ddot{x}^i$ is non-trivial even in the absence of forcing.  This fact simplifies the analysis when we study the small mass limit of the noisy system and is one of the advantages of the orthogonal frame bundle formulation.

Using the above lemma, we can write the equation for the integral curves of $X$, \req{X_def}, in coordinates.
\begin{corollary}
 In a coordinate system defined as in Lemma \ref{X_coords_lemma}, an integral curve of $X$, $(x^i(t), h^\beta_\delta(t),v^\alpha(t))$, satisfies
\begin{align}
\dot{x}^j(t)=&(\Lambda^{-1})_\alpha^j(x(t))h_\beta^\alpha(t)v^\beta(t),\label{coord_odes_x}\\
\dot{h}_\beta^\alpha(t)=&-h_\beta^\eta(t) A_{\xi\eta}^\alpha(x(t))h_\delta^\xi(t)v^\delta(t),\label{coord_odes_h}\\
\dot{v}^\alpha(t)=&(h^{-1})^\alpha_\beta(t)V^\beta(\dot{x}(t))\label{coord_odes_v},
\end{align}
where, from \req{tau_clarif}, we see that $\dot{x}(t)=v^\xi(t) h_\xi^\eta(t)  E_\eta(x(t))$. Recall that $V^\beta$ are the components of $V$ in the o.n. basis $E_\beta$.
\end{corollary}

In the process of proving Lemma \ref{X_coords_lemma} we have also characterized the relation between integral curves of $X$, \req{X_def}, and integral curves of $Y$, \req{Y_def}, as expressed by the following corollaries.
\begin{corollary}
Let $x(t)$ be the  solution to 
\begin{equation}
\nabla_{\dot{x}}\dot{x}=V(\dot{x}),\hspace{2mm} x(0)=\pi(u),\hspace{2mm}\dot{x}(0)=u(v),
\end{equation}
 i.e. the  integral curve of $Y$, defined in \req{Y_def}, starting at $u(v)\in T_{\pi(u)}M$.

  Let $U_\alpha(t)$ be the parallel translates of $u(e_\alpha)$ along $x(t)$ ($e_\alpha$ is the standard basis for $\mathbb{R}^n$), i.e.
\begin{align}
\nabla_{\dot{x}}U_\alpha=0,\hspace{2mm} U_\alpha(0)=u(e_\alpha).
\end{align}
  Parallel transport preserves inner products, so $\tau(t)$ defined by $\tau(t)e_\alpha=U_\alpha(t)$ is a smooth section of $F_O(M)$ along $x(t)$.  Define the smooth curve in $\mathbb{R}^n$, $v(t)=\tau(t)^{-1}\dot{x}(t)$.

 Define the smooth curve in $N$, $\eta(t)=(\tau(t),v(t))$. This is  an integral curve of $X$ starting at $(u,v)$.

\end{corollary}
Conversely, uniqueness of integral curves gives us the following.
\begin{corollary}
Let $(\tau(t),v(t))$ be an integral curve of $X$ starting at $(u,v)$.  Define $x(t)=\pi(\tau(t))$ and $U_\alpha(t)=\tau(t)e_\alpha$.  Then $x(t)$ is a solution to
\begin{equation}
\nabla_{\dot{x}}\dot{x}=V(\dot{x}),\hspace{2mm} x(0)=\pi(u),\hspace{2mm}\dot{x}(0)=u(v),
\end{equation}
the $U_\alpha$ are parallel  along $x(t)$, and $v(t)=\tau^{-1}(t)\dot{x}(t)$.
\end{corollary}

\subsection{A Second Coordinate Independent Formulation}
In this section, we introduce a natural set of horizontal vector fields on the orthogonal frame bundle and, using the coordinate expression for $X$, \req{X_coords}, we show that these vector fields can be used to characterize the dynamical system \req{ode_system}, yielding another coordinate independent formulation.  This will also show the relationship between the equations \reqr{coord_odes_x}{coord_odes_v}  and the equations in \cite{XueMei2014}.  The formulation we give in this section will be utilized for the remainder of the paper, as it has several advantages over our previous characterizations of the system \req{ode_system}. These advantages will be made clear as we progress.

\begin{lemma}\label{horizontal_vf}
On an $n$-dimensional Riemannian manifold there exists a canonical linear map from $\mathbb{R}^n$ to horizontal vector fields on $F_O(M)$ defined as follows (see \cite{hsu2002stochastic,XueMei2014}).

For each $v\in\mathbb{R}^n$ and $u\in F_O(M)$, define $H_v(u)\in T_u F_O(M)$ by $H_v(u)=(u(v))^h$, i.e. the horizontal lift of $u(v)\in T_{\pi(u)} M$ to $T_uF_O(M)$.

  This is a smooth horizontal vector field on $F_O(M)$.   Pushing forward to $U\times O(\mathbb{R}^n)$ via a  local trivialization $(U,\Phi)$ of $F_O(M)$ about $u$ with corresponding o.n. frame $E_\alpha$, as in Section \ref{sec:X_def1}, they have the  form
 \begin{equation}\label{H_def}
H_v(u)=v^\alpha h^\beta_\alpha E_\beta(\pi(u)) -v^\alpha h^\beta_\alpha h^\eta_\xi A_{\beta \eta}^\delta(\pi(u))\partial_{e^\delta_\xi},
\end{equation}
where $\Phi(u)=(\pi(u),h)$, $\nabla_{E_\alpha}E_\beta=A^\eta_{\alpha \beta}E_\eta$, $v^\alpha$ are the components of $v$ in the standard basis for $\mathbb{R}^n$, and $e_\alpha^\beta$ are the standard coordinates on $\mathbb{R}^{n\times n}$.  Note that the second term defines a vector field on $\mathbb{R}^{n\times n}$, but it is in fact tangent to $O(\mathbb{R}^n)$. Our expression \req{H_def} differs slightly from the one found in \cite{hsu2002stochastic}, as we have written it in an o.n. frame rather than a coordinate frame.

If $e_\alpha$ is the standard basis for $\mathbb{R}^n$ we will let $H_\alpha\equiv H_{e_\alpha}$. Therefore, $H_v=v^\alpha H_\alpha$  for any $v\in \mathbb{R}^n$, where we employ the summation convention. 

Under right multiplication by $g\in O(\mathbb{R}^n)$, these vector fields satisfy
\begin{align}
(R_g)_*(H_v(u))=H_{g^{-1}v}(u g).
\end{align}
\end{lemma}

\begin{remark}
The implied summations in $v^\alpha h_{\alpha}^\beta$, $v^\alpha H_\alpha$, etc., are summations over components in the standard basis for $\mathbb{R}^n$.  The $\alpha$'s here are not tensor indices on $M$, $TM$, or $F_O(M)$ and do not transform under change of coordinates or frame.  This is in contrast with the index $\beta$ in $h_\alpha^\beta$, which does transform under a change of the o.n. frame $E_\beta$.  We will occasionally revisit this point going forward for emphasis.
\end{remark}

The horizontal vector fields \req{H_def} can be used to relate geodesic motion and parallel transport on $M$ to a flow on the frame bundle.
\begin{lemma}
Let $u\in F_O(M)$ and $v\in\mathbb{R}^n$.  Let $\tau$ be the integral curve of $H_v$ starting at $u$.  Then $x\equiv \pi\circ\tau$ is the geodesic starting at $\pi(u)$ with initial velocity $u(v)$ and for any $w\in\mathbb{R}^n$, $\tau(t) w$ is  parallel transported along $x(t)$.
\end{lemma}
\begin{proof}
$\tau$ is a horizontal curve in $F_O(M)$ iff $\tau(w)$ is horizontal in $TM$ for any $w\in\mathbb{R}^n$.  In a vector bundle, horizontal and parallel transported are synonymous. Hence $\tau(t)w$ is parallel transported along $x=\pi\circ \tau$. Therefore, to prove $x(t)$ is the claimed geodesic it suffices to show $\dot{x}=\tau(v)$.

In a local trivialization $\Phi$, $\Phi(\tau(t))=(x(t),h(t))$. Hence, using \req{H_def}, we have
\begin{align}
\dot{x}(t)=v^\alpha h^\beta_\alpha(t) E_\beta(\pi(\tau(t)))=\tau(t) v.
\end{align}

\end{proof}
Uniqueness of geodesics, parallel transport, and integral curves then gives the following.
\begin{lemma}
Let $x\in M$, $u$ be a frame at $x$, and $v\in\mathbb{R}^n$.  Let $x(t)$ be the geodesic starting at $x$ with initial velocity $u(v)$.  Let $e_\alpha$ be the standard basis for $\mathbb{R}^n$ and $U_\alpha$ be the parallel translates of $u(e_\alpha)$ along $x(t)$.  Let $\tau(t)$ be the corresponding section of $F_O(M)$, i.e. $\tau(t) e_\alpha=U_\alpha(t)$.  Then $\tau$ is the integral curve of $H_v$ starting at $u$.
\end{lemma}

We can also use the $H$'s to lift vector fields from $M$ to the frame bundle.
\begin{lemma}\label{horiz_lift_vf_formula}
Let $b$ be a smooth vector field on $M$ and $b^h$ be the horizontal lift of $b$ to $F_O(M)$. Recall that this is a smooth vector field on $F_O(M)$.  We have
\begin{align}
b^h(u)= H_{u^{-1}b(\pi(u))}(u).
\end{align}
If $R_g$ denotes right multiplication by $g\in O(\mathbb{R}^n)$ then $(R_g)_*b^h=b^h$.
\end{lemma}
\begin{proof}
To prove the first assertion, by the definition of $H$,
\begin{align}
H_{u^{-1}b(\pi(u))}(u)= (u(u^{-1}b(\pi(u))))^h=(b(\pi(u)))^h=b^h(u).
\end{align}
As for the second,
\begin{align}
\pi_*( ((R_g)_* b^h)(u))=(\pi\circ R_g)_* b^h(u g^{-1})=b(\pi(ug^{-1}))=b(\pi(u)).
\end{align}
$(R_g)_*$ preserves the horizontal subspaces, hence $(R_g)_*b^h$  is the horizontal lift of $b$.
\end{proof}

\begin{lemma}
Let $b$ be a smooth vector field on $M$. If $\tau$ is an integral curve of $b^h$ starting at $u$ then $x\equiv\pi\circ\tau$ is an integral curve of $b$ starting at $\pi(u)$ and for any $v\in \mathbb{R}^n$, $\tau(t)v$ is the parallel translate of $u(v)$ along $x(t)$.

Conversely, if $x(t)$ is an integral curve of $b$ starting at $\pi(u)$ and $U_\alpha(t)$ are the parallel translates of $u(e_\alpha)$ along $x(t)$ then $\tau(t)$ defined by $\tau(t)e_\alpha=U_\alpha(t)$ is the integral curve of $b^h$ starting at $u$.
\end{lemma}
\begin{proof}
Suppose $\tau$ is an integral curve of $b^h$ starting at $u$.  Then
\begin{align}
\dot{x}=\pi_*\dot{\tau}=\pi_* b^h(\tau)=b(x).
\end{align}
So $x(t)$ is an integral curve of $b$.  $\tau$ has horizontal tangent vector for all $t$, hence $\tau(t) v$ is parallel in $TM$.

Conversely, if $x(t)$ is an integral curve of $b$ starting at $\pi(u)$ and $U_\alpha(t)$ are the parallel translates of $u(e_\alpha)$ then $\tau(t)$ defined by  $\tau(t)e_\alpha=U_\alpha(t)$ is a smooth horizontal curve in $F_O(M)$ and $\tau(t_0)=u$.  We have
\begin{align}
\pi_*\dot{\tau}=\dot{x}=b(x).  
\end{align}
$\dot{\tau}$ is horizontal, so 
\begin{align}
\dot{\tau}=(b(x))^h=b^h(\tau).
\end{align}
\end{proof}

The prior lemmas show that geodesic motion, parallel transport, and flows on $M$ can all be related to flows on $F_O(M)$.  Therefore, it shouldn't come as a surprise that the vector field $X$, \req{X_def}, whose integral curves characterize the trajectories of our deterministic system, can be written in terms of the $H_v$'s and the forcing, $V$.
\begin{proposition}\label{X_H_relation_prop}
Let $N= F_O(M)\times \mathbb{R}^n$ and $(u,v)\in N$.  Then $X_{(u,v)}$, defined by \req{X_def}, is given by
\begin{align}\label{X_H_relation}
X_{(u,v)}=(H_v(u),u^{-1}V(u(v)))
\end{align}
where we have identified $T\mathbb{R}^n$ with $\mathbb{R}^n$.
\end{proposition}
\begin{proof}
In a local trivialization induced by an o.n. frame $E_\alpha$, \req{X_coords} implies that
\begin{align}
X|_{(x,h,v)}=&h_\beta^\alpha v^\beta E_\alpha(x)-h_\alpha^\eta A_{\delta\eta}^\beta(x)h_\xi^\delta v^\xi\partial_{e^\beta_\alpha}+(h^{-1})^\alpha_\beta V^\beta(u(v))\partial_{v^\alpha}.
\end{align}
The proposition then follows from \req{H_def}.
\end{proof}
The geometric significance of the $H_v$'s will make \req{X_H_relation} simpler to work with than our initial definition of the vector field $X$, \req{X_def}.

Proposition \ref{X_H_relation_prop} implies that the deterministic dynamics of the system of interest, \req{ode_system}, lifted to  $N=F_O(M)\times \mathbb{R}^n$, are given by
\begin{align}\label{determ_system}
\dot{u}= H_v(u),\hspace{2mm} \dot{v}=u^{-1}V(u(v)), \hspace{2mm} (u(t_0),v(t_0))=(u_0,v_0).
\end{align} We want to emphasize that $v$ is defined in terms of the dynamical frame $u$, and not in reference to any choice of coordinates on $M$ or $F_O(M)$.  In other words, the components $v^\alpha$ of $v$ in the standard basis for $\mathbb{R}^n$ are the components of the particle's velocity {\em in its own parallel transported frame}. They are not tied to a particular coordinate system on $M$ or $F_O(M)$ and do not transform under coordinate changes on either  space.

\section{Randomly Perturbed Geodesic Flow With Forcing}\label{sec:SDE_manifold}
In this section we will  show how we couple noise to the system \req{determ_system} to obtain a stochastic differential equation on $N$. 
\subsection{Stochastic Differential Equations on Manifolds}\label{SDE_manifold_sec}
First we recall the definition and some basic properties of semimartingales and stochastic differential equations  on manifolds.    The definition and lemmas in this section are adapted from \cite{hsu2002stochastic}, but we repeat them here for completeness. The general theory outlined in this section does not require a Riemannian metric on $M$.

\begin{definition}
Let $M$ be an $n$-dimensional smooth manifold, $(\Omega,\mathcal{F},\mathcal{F}_t,P)$ be a filtered probability space satisfying the usual conditions \cite{karatzas2014brownian}, and $X_t$ be a continuous adapted $M$-valued process.  $X$ is called an  $M$-valued continuous semimartingale if $f\circ X_t$ is an $\mathbb{R}$-valued semimartingale for all $f\in C^\infty(M)$.   We will only deal with continuous semimartingales, so we drop the adjective continuous from now on.
\end{definition}
Note that, by It\^o's formula,  if $M=\mathbb{R}^n$ then this agrees with the usual definition.

\begin{definition}\label{def:manifold_SDE}
Let $V$ be a $k$-dimensional vector space and $Z_t$ be a $V$-valued semimartingale, called the driving process.  Let $M$ be a smooth manifold, $X_t$ be an $M$-valued semimartingale, and $\sigma$ be a smooth section of $TM\bigotimes V^*$. We say that $X_t$ is a solution to the SDE
\begin{align}\label{manifold SDE}
X_t=X_{t_0}+\int_{t_0}^t \sigma(X_s)\circ dZ_s
\end{align}
if
\begin{align}\label{manifold_SDE_def}
f(X_t)=f(X_{t_0})+\int_{t_0}^t\sigma(X_s) [f]\circ dZ_s
\end{align}
$P$-a.s. for all $f\in C^\infty(M)$, where $\int ...\circ dZ_s$ denotes the stochastic integral in the Stratonovich sense.  We use the notation $Y[f]$ to denote the smooth function one obtains by operating with some vector field, $Y$, on a smooth function, $f$, and in the stochastic integral we contract over the $V^*$ and $V$ factors from $\sigma[f]$ and $Z$ respectively.  We will equivalently write the SDE \req{manifold SDE} in differential notation
\begin{align}\label{manifold SDE}
dX_t=\sigma(X_s)\circ dZ_s.
\end{align}
\end{definition}
Note that when $M$ is a finite dimensional vector space, this definition agrees with the usual one (in the Stratonovich sense). Using a basis for $V$ and the dual basis for $V^*$ to write the contraction in \req{manifold_SDE_def} as a sum over components in these bases we arrive at a formula analogous to the definition in \cite{hsu2002stochastic} (page 21). However, we find it useful to use the above formulation in terms of a vector space and its dual in order to justify use of the summation convention over contracted indices.

The Stratonovich integral is used in \req{manifold_SDE_def} in order to make the definition  diffeomorphism-invariant, as captured by the following Stratonovich calculus variant of the It\^o change-of-variables formula (see \cite{hsu2002stochastic} pp.20-21).
\begin{lemma}\label{SDE_diffeo_lemma}
Let $X_t$ be an $M$-valued semimartingale that satisfies the SDE 
\begin{align}
X_t=X_{t_0}+\int_{t_0}^t \sigma(X_s)\circ dZ_s,
\end{align}
$N$ be another smooth manifold, and $\Phi:M\rightarrow N$ be a diffeomorphism.  Then $\tilde X\equiv \Phi\circ X$ is an $N$-valued semimartingale and satisfies the SDE
\begin{align}
\tilde X_t=\tilde X_{t_0}+\int_{t_0}^t (\Phi_*\sigma)(\tilde X_s)\circ dZ_s
\end{align}
where $\Phi_*$ denotes the pushforward.
\end{lemma}
Definition \ref{manifold_SDE_def} can be restated in terms of the It\^o integral as follows, similar to p.23 of \cite{hsu2002stochastic}.
\begin{lemma}\label{Ito_SDE_lemma}
 $X_t$ is a solution to  the SDE
\begin{align}
X_t=X_{t_0}+\int_{t_0}^t \sigma(X_s)\circ dZ_s
\end{align}
iff
\begin{align}
f(X_t)=f(X_{t_0})+\int_{t_0}^t\sigma(X_s) [f] dZ_s+\frac{1}{2}\int_{t_0}^t \sigma_\alpha(X_s)[\sigma_\beta[f]] d[Z^\alpha,Z^\beta]_s
\end{align}
for all $f\in C^\infty(M)$ where the summation convention is employed and the sum is over the components in any basis, dual basis pair for $V$ and $V^*$.  This is another manifestation of the  It\^o formula for the stochastic differential of the composition of a smooth function with a semimartingale.  
\end{lemma}

\subsection{Coupling to Noise}
 For the remainder of this paper we will assume $M$  is compact, connected, and without boundary.  Note that this also implies $F_O(M)$ is compact and without boundary.    In this section we describe the coupling of the dynamical system \req{determ_system} to noise, and hence we must also assume $M$ is equipped with a Riemannian metric.

 Let $W$ be an $\mathbb{R}^k$-valued Wiener process and $\sigma:F_O(M)\rightarrow \mathbb{R}^{n\times k}$ be smooth.  We are interested in the following SDE for $(u,v)\in N= F_O(M)\times \mathbb{R}^n$,
\begin{align}
u_t=&u_0+\int_{t_0}^t H_{v_s}(u_s)ds,\label{SDE_u_m}\\
v_t=&v_0+\frac{1}{m}\int_{t_0}^t u_s^{-1}V(u_sv_s)ds+\frac{1}{m}\int_{t_0}^t \sigma(u_s)\circ dW_s.\label{SDE_v_m}
\end{align}
Note that we have replaced $V$ in  \req{determ_system}  with $\frac{1}{m} V$ (where now $V$ is independent of $m$), making the dependence on particle mass, $m$, explicit.

To connect with  Definition \ref{def:manifold_SDE}, one must view $H_v(u)$, $\frac{1}{m} u^{-1}V(uv)$, and $\frac{1}{m}\sigma(u)$ as sections of $TN\bigotimes (\mathbb{R}^{k+1})^*$ (identifying $T_{(u,v)}N$ with $T_uF_O(M)\bigoplus \mathbb{R}^n$), and use the driving $\mathbb{R}^{k+1}$-valued semimartingale  $Z_t=(t,W_t)$. Alternatively, one could view the above objects as $k+1$ vector fields on $N$ and include sums over indices, as done in \cite{hsu2002stochastic}, but for economy of notation, we wish to avoid employing indices and explicit summations when possible.

Because the Wiener process only couples to the equation for $v$, which is a process with values in the second factor of the product space $N= F_O(M)\times \mathbb{R}^n$, a solution of the SDE \reqr{SDE_u_m}{SDE_v_m} on the manifold $N$ in the sense of \req{manifold_SDE_def} is equivalent to the existence of an $N$-valued semimartingale, $(u,v)$, such that the first component is pathwise $C^1$ and pathwise satisfies the ODE
\begin{align}\label{u_ode}
\dot{u}_t=H_{v_t}(u_t), \hspace{2mm} u(t_0)=u_0
\end{align}
and the second component satisfies the SDE on $\mathbb{R}^n$
\begin{align}
v_t=v_0+\frac{1}{m}\int_{t_0}^t u_s^{-1}V(u_sv_s)ds+\frac{1}{m}\int_{t_0}^t \sigma(u_s) dW_s.
\end{align}
Note that $u$ has locally bounded variation, so the choice of stochastic integral in the second equation is not significant.  We use the It\^o notation here.  We emphasize that, while the machinery of Section \ref{SDE_manifold_sec} is not needed in order to formulate the above system, it will be required when we pass to the limit $m\rightarrow 0$.  

For the remainder of the paper we will make the following assumption.
\begin{assumption}\label{drag_assump}
 We will assume that the deterministic vector field $V$ is the sum of a position dependent force term and a position dependent linear drag term
\begin{align}\label{assump_V_form}
V(w)=F(x)-\gamma(x) w, \hspace{2mm}w\in T_xM,\hspace{2mm} x=\pi(w)\in M,
\end{align}  
where  $F$ is a smooth vector field on $M$ and $\gamma$ is a smooth $\binom{1}{1}$ tensor field on $M$. We will {\em not} assume that the force field $F$ comes from a potential.  

As stated in Section \ref{sec:result_summary}, we will also assume that the symmetric part of $\gamma$, $\gamma^s=\frac{1}{2}(\gamma+\gamma^T)$, has eigenvalues bounded below by a constant $\gamma_1>0$  on all of $M$. We again emphasize that this coercivity assumption will be crucial for the momentum decay estimates of Section  \ref{sec:p_zero}.
\end{assumption}

 In the following it will be useful to denote $u^{-1}F(\pi(u))$ by $F(u)$ and $u^{-1}\gamma(\pi(u)) u$ by $\gamma(u)$, letting the context distinguish between the different notations.  These are smooth $\mathbb{R}^n$ and $\mathbb{R}^{n\times n}$-valued functions on $F_O(M)$ respectively. With these definitions, the SDE \req{SDE_u_m} - \req{SDE_v_m}  becomes
\begin{align}
u_t=&u_0+\int_{t_0}^t H_{v_s}(u_s)ds,\label{SDE_u_m2}\\
v_t=&v_0+\frac{1}{m}\int_{t_0}^t [F(u_s)-\gamma(u_s)v_s]ds+\frac{1}{m}\int_{t_0}^t \sigma(u_s) dW_s.\label{SDE_v_m2}
\end{align}
Given $k$ vector fields, $\sigma_\alpha(x)$, on $M$, these induce corresponding noise coefficients on the frame bundle, $\sigma(u)$,  given by
\begin{align}\label{k_vf_case}
\sigma(u)e_\alpha=u^{-1}\sigma_\alpha(\pi(u)).
\end{align}
Additionally, one is often interested in the case where $k=n$ and $\sigma(u)$  comes from a $\binom{1}{1}$-tensor field on a $M$, denoted $\sigma(x)$, in the same manner as $\gamma(u)$ i.e.
\begin{align}\label{sigma_tensor}
\sigma(u)=u^{-1}\sigma(\pi(u))u.
\end{align}
For most of this work we keep the discussion general and deal only with $\sigma(u)$.

The following lemmas will be useful.
\begin{lemma}
Let $\gamma^s$ denote the symmetric part of $\gamma$.  Then
\begin{align}
\gamma^s(u)=u^{-1}\gamma^{s}(\pi(u))u
\end{align}

\end{lemma}
\begin{proof}
We are done if we can show
\begin{align}
\gamma^T(u)=u^{-1}\gamma^T(\pi(u))u.
\end{align}
Letting $\cdot$ be the Euclidean inner product on $\mathbb{R}^n$, for $x,y\in\mathbb{R}^n$ we have
\begin{align}
&y\cdot\gamma^T(u)x=(\gamma(u)y)\cdot x=(u^{-1}\gamma(\pi(u))uy)\cdot x=g(\gamma(\pi(u))uy,u x)\\
=&g(uy,\gamma^T(\pi(u))u x)=y\cdot (u^{-1}\gamma^T(\pi(u))u x).\notag
\end{align}
This holds for all $x,y$ and so the proof is complete.
\end{proof}
\begin{corollary}\label{exp_gamma_bounds}
The eigenvalues of $\gamma^s(u)$ and $\gamma^s(\pi(u))$ are the same.  In particular, by Assumption \ref{drag_assump}, the eigenvalues of $\gamma^s(u)$ are also bounded below by $\gamma_1>0$ for all $u\in F_O(M)$.

This also implies that the real parts of the eigenvalues of $\gamma(u)$ are bounded below by $\gamma_1$ for all $u\in F_O(M)$. In addition,
\begin{align}
\|e^{-t\gamma(u)}\|\leq e^{-\gamma_1 t},\hspace{ 2mm} \|e^{-t\gamma(u)^T}\|\leq e^{-\gamma_1 t},
\end{align}
for any $u\in F_O(M)$ and any $t\geq 0$  (see, for example, p.86 of \cite{teschl2012ordinary}).
\end{corollary}

\begin{lemma}\label{m_existence_lemma}
For each $(u_0,v_0)\in N$ there exists a unique globally defined solution $(u_t,v_t)$, $t\in[0,\infty)$ to the SDE \ref{SDE_u_m2} - \ref{SDE_v_m2} that pathwise satisfies the initial conditions.  It can be chosen so that pathwise, $t \mapsto u_t$ is $C^1$ and satisfies the ODE \ref{u_ode}.  We emphasize that the global in time existence relies on the compactness of $M$.
\end{lemma}
\begin{proof}
The diffusion term for the SDE is independent of $v$ and the drift is an affine function of $v$, so compactness of $F_O(M)$ implies that the drift and diffusion are linearly bounded in $v$, uniformly in $u$.  Therefore, by embedding $F_O(M)$ compactly in some $\mathbb{R}^l$, one can use the results on global existence and uniqueness of solutions to a vector-valued SDE with linearly bounded coefficients (see for example \cite{karatzas2014brownian}, Theorem 5.2.9) to prove existence of a unique globally defined solution to the SDE that pathwise satisfies the initial conditions. One can modify the result on a measure zero set to ensure that  the $u$ component is also a $C^1$-function of $t$ and satisfies the ODE \req{u_ode} everywhere, not just a.s.
\end{proof}

Often one is only interested in the evolution of the position, $x_t=\pi(u_t)$, and velocity, $\dot{x}_t$,  degrees of freedom.  The SDE \req{SDE_u_m2} implies that $\dot{x}_t=u_tv_t$ and, pathwise, $u(t)$ is horizontal.  In particular, for any $w\in\mathbb{R}^n$, $u(t)w$ is parallel transported along $x(t)$, the same as for the deterministic system.  The following lemma captures the dependence of the solution on the choice of an initial frame in the case where $\sigma$ is given by \req{sigma_tensor}.
\begin{lemma}
Let $h\in O(\mathbb{R}^n)$ and $(u_t,v_t)$ be the solution to \reqr{SDE_u_m2}{SDE_v_m2}  corresponding to the initial condition $(u_0,v_0)$.  Suppose  $\sigma(u)$ is obtained from $\sigma(x)$ as in $\req{sigma_tensor}$. Then
\begin{align}
(\tilde u_t,\tilde v_t)\equiv (u_t h,h^{-1}v_t)
\end{align}
is the solution to \reqr{SDE_u_m2}{SDE_v_m2}   with the initial condition $(u_0 h,h^{-1}v_0)$ and the Wiener process $W_t$ replaced by the Wiener process $\tilde W_t=h^{-1} W_t$.
\end{lemma}
\begin{proof}
$(\tilde u_t,\tilde v_t)$ is  a semimartingale starting at $(u_0 h,h^{-1}v_0)$.   The map 
\begin{equation}
\Phi(u,v)= (u h,h^{-1}v)
\end{equation}
 is a diffeomorphism of $N$ and therefore Lemma \ref{SDE_diffeo_lemma} implies 
\begin{align}   
&d\tilde u_t= (R_{h})_*(H_{v(t)}(u_t)) dt,\hspace{2mm} d\tilde v_t\\
=&\frac{1}{m}(L_{h^{-1}})_*(F(u_t)-\gamma(u_t)v_t)dt+\frac{1}{m} (L_{h^{-1}})_*\sigma(u_s) dW_s\notag
\end{align}
where $R$ and $L$ denote right and left multiplication respectively.  Using the definitions of $F(u)$, $\gamma(u)$ and $\sigma(u)$ this simplifies to 
\begin{align}   
d\tilde u_t=H_{\tilde v_t}(\tilde u_t) dt,\hspace{2mm} d\tilde v_t=\frac{1}{m}((F(\tilde u_t)-\gamma(\tilde u_t)\tilde v_t)dt+\frac{1}{m} \sigma(\tilde u_t) d\tilde W_t.
\end{align}
\end{proof}

\section{Rate of Decay of the Momentum}\label{sec:p_zero}
We now begin our investigation of the properties of the solutions of the SDE \ref{SDE_u_m2} - \ref{SDE_v_m2}  in the small mass limit by proving that the momentum process, $p_t=mv_t$, converges to zero in several senses as $m\rightarrow 0$.  To this end, we will introduce a superscript to the solutions, $(u^m_t,v^m_t)$, of \reqr{SDE_u_m2}{SDE_v_m2}  to denote the corresponding value of the mass.  The non-random initial conditions, $u_0,v_0$, will be fixed independently of $m$.

More specifically, the momentum process will be shown to converge to zero with a rate dependent on powers of $m$. This convergence is shown with respect to the uniform $L^p$-metric on continuous paths (Prop. \ref{E_sup_p_bound}), $L^p$ metric (Prop. \ref{sup_E_p_bound}), and as a stochastic integral with respect to the momentum (Prop. \ref{quad_p_int_bound}). To prove these propositions, the equation for $v^m_t$, \req{SDE_v_m2}, is solved in terms of $u^m_t$. Estimates are made on the Lebesgue integrals much like in the ordinary differential equation case. The stochastic integral term is rewritten  in order to mirror the ODE case as closely as possible and then broken into small intervals which can be controlled using the Burkholder-Davis-Gundy inequalities.

First we give some useful lemmas. 
\subsection{Some Lemmas}

\begin{lemma}\label{matrix_exp_decay_bound}
Let $X_t=X_0+M_t+A_t$ be a continuous $\mathbb{R}^k$-valued semimartingale on $(\Omega,\mathcal{F},\mathcal{F}_t,P)$ with local martingale component $M_t$ and locally bounded variation component $A_t$. Let $V\in L^1_{loc}(A)\cap L^2_{loc}(M)$ be $\mathbb{R}^{n\times k}$-valued and  let $B(t)$ be a continuous $\mathbb{R}^{n\times n}$-valued adapted process.  Let $\Phi(t)$ be the adapted $\mathbb{R}^{n\times n}$-valued  $C^1$ process that pathwise solves
the initial value problem (IVP) 
\begin{align}
\dot{\Phi}(t)=B(t)\Phi(t),\hspace{2mm} \Phi(0)=I.
\end{align}
Then we have the $P$-a.s. equalities
\begin{align}
&\Phi(t)\int_0^t\Phi^{-1}(s) V_s dX_s=\int_0^t V_s dX_s+\Phi(t)\int_0^t   \Phi^{-1}(s) B(s)\left(\int_0^s V_r dX_r\right)ds\label{Phi_int_equality}\\
=&\Phi(t)\int_0^t V_s dX_s-\Phi(t)\int_0^t \Phi^{-1} (s)B(s) \left(\int_s^t V_r dX_r\right) ds\label{Phi_int_equality2}
\end{align}
for all $t$.

If the eigenvalues of the symmetric part of $B$, $B^s=\frac{1}{2}(B+B^T)$, are bounded above by $-\alpha$ for some $\alpha>0$ then for every $T\geq\delta>0$ we have the $P$-a.s. bound
\begin{align}\label{Phi_int_bound2}
&\sup_{t\in[0,T]}\|\Phi(t)\int_0^t\Phi^{-1}(s) V_s dX_s\|\leq (1+\frac{4}{ \alpha}\sup_{s\in[0,T]}\|B(s)\|)\left(e^{-\alpha\delta}\sup_{t\in[0,T]}\|\int_0^t V_r dX_r\|\right. \notag\\
&\left.+\max_{k=0,...,N-1}\sup_{t\in [k\delta, (k+2)\delta]}\|\int_{k\delta}^tV_rdX_r\|\right)
\end{align}
where $N=\max\{k\in\mathbb{Z}:k\delta<T\}$.  Here and in the following we use the $\ell^2$ norm on every $\mathbb{R}^k$.
\end{lemma}
\begin{proof}
Using integration by parts, together with the fact that $\Phi$ is a process of locally bounded variation and $\dot{\Phi}(t)=B(t)\Phi(t)$, we obtain the $P$-a.s. equality
\begin{align}\label{P_as_int_eq}
\Phi(t)\int_0^t \Phi^{-1}(s)V_s dX_s=&\int_0^t V_sdX_s+\int_0^t B(s)\Phi(s) \int_0^s \Phi^{-1}(r)V_r dX_r ds
\end{align}
for all $t$.

 Fix an $\omega\in\Omega$ for which the above equality holds and consider the resulting continuous functions $r(t)=\int_0^t V_s dX_s$ and $y(t)=\Phi(t)\int_0^t\Phi^{-1}(s) V_s dX_s$.  \req{P_as_int_eq}  implies that these satisfy the integral equation
\begin{align}
y(t)=r(t)+\int_0^t B(s) y(s)ds,\hspace{2mm} y(0)=0.
\end{align}
The unique solution to this equation is \cite{kedem1981posteriori}
\begin{align}
y(t)=r(t)+\Phi(t) \int_0^t  \Phi^{-1}(s) B(s) r(s) ds.
\end{align}
This proves the first equality in \req{Phi_int_equality}.  For the second, we compute
\begin{align}
&\Phi(t)\int_0^t V_s dX_s-\Phi(t)\int_0^t \Phi^{-1} (s)B(s) \left(\int_s^t V_r dX_r\right) ds\\
=&\Phi(t)\left(\int_0^t V_s dX_s-\int_0^t \Phi^{-1} (s)B(s) \left(\int_0^t V_r dX_r -\int_0^s V_r dX_r\right) ds\right)\notag\\
=&\Phi(t)\int_0^t \Phi^{-1} (s)B(s)\left(\int_0^s V_r dX_r\right) ds+\Phi(t)\left(I-\int_0^t \Phi^{-1} (s)B(s) ds\right) \int_0^t V_r dX_r\notag\\
=&\Phi(t)\int_0^t \Phi^{-1} (s)B(s)\left(\int_0^s V_r dX_r\right) ds+\Phi(t)\left(I+\int_0^t \frac{d}{ds}\Phi^{-1} (s) ds\right) \int_0^t V_r dX_r\notag\\
=&\Phi(t)\int_0^t \Phi^{-1} (s)B(s)\left(\int_0^s V_r dX_r\right) ds+ \int_0^t V_r dX_r,\notag
\end{align}
where we have used the formula
\begin{equation}
\frac{d}{ds}\Phi^{-1}(s) = -\Phi^{-1}(s)\dot{\Phi}(s)\Phi^{-1}(s).
\end{equation}

To obtain the bound \req{Phi_int_bound2} we start from \req{Phi_int_equality2} and take the norm to find
\begin{align}
&\|\Phi(t)\int_0^t\Phi^{-1}(s) V_s dX_s\|\\
\leq &\|\Phi(t)\|\|\int_0^t V_s dX_s\|+\int_0^t \|\Phi(t)\Phi^{-1} (s)\|\|B(s)\|\|\int_s^t V_r dX_r\| ds.\notag
\end{align}
For $t\geq s$, the fundamental solution $\Phi(t) \Phi^{-1}(s)$ satisfies the bound
\begin{align}\label{fundamental_matrix}
\|\Phi(t) \Phi^{-1}(s)\|\leq e^{\int_{s}^t \lambda_{max}(r)dr}
\end{align}
where $\lambda_{max}(r)$ is the largest eigenvalue of $B^s(r)$ (see, for example, p.86 of \cite{teschl2012ordinary}).  Therefore, assuming $\lambda_{\max}\leq-\alpha<0$ gives
\begin{align}
&\|\Phi(t)\int_0^t\Phi^{-1}(s) V_s dX_s\|\\
\leq & e^{-\alpha t}\|\int_0^t V_s dX_s\|+\sup_{s\in[0,t]}\|B(s)\|\int_0^t e^{-\alpha(t-s)}\|\int_s^t V_r dX_r\| ds.\notag
\end{align}
For any $T\geq \delta>0$ we have the $P$-a.s. bounds
\begin{align}
\sup_{t\in[0,T]}e^{-\alpha t}\|\int_0^t V_s dX_s\|\leq \sup_{t\in[0,\delta]}\|\int_0^t V_s dX_s\|+e^{-\alpha \delta}\sup_{t\in[\delta,T]}\|\int_0^t V_s dX_s\|
\end{align}
and
\begin{align}
&\sup_{t\in[0,T]}\left(\sup_{s\in[0,t]}\|B(s)\|\int_0^t e^{-\alpha(t-s)}\|\int_s^t V_r dX_r\| ds\right)\notag\\
\leq &\sup_{s\in[0,T]}\|B(s)\|\left(\sup_{t\in[0,\delta]}\int_0^t e^{-\alpha(t-s)}\|\int_s^t V_r dX_r\| ds\right.\\
&\left.+\sup_{t\in[\delta,T]}\int_0^t e^{-\alpha(t-s)}\|\int_s^t V_r dX_r\| ds\right).\notag
\end{align}
The first term can be bounded as follows.
\begin{align}\label{first_term_bound}
&\sup_{t\in[0,\delta]}\int_0^t e^{-\alpha(t-s)}\|\int_s^t V_r dX_r\| ds\\
=&\sup_{t\in[0,\delta]}\int_0^t e^{-\alpha(t-s)}\|\int_0^t V_r dX_r-\int_0^s V_r dX_r\| ds\notag\\
\leq&\sup_{t\in[0,\delta]}\int_0^t e^{-\alpha(t-s)}2\sup_{0\leq \tau\leq\delta}\|\int_0^\tau V_r dX_r\|ds\notag
\leq\frac{2}{\alpha}\sup_{0\leq t\leq\delta}\|\int_0^t V_r dX_r\|.
\end{align}
In the second term we can split the integral to obtain
\begin{align}
&\sup_{t\in[\delta,T]}\int_0^t e^{-\alpha(t-s)}\|\int_s^t V_r dX_r\| ds\\
=&\sup_{t\in[\delta,T]}\left(\int_0^{t-\delta} e^{-\alpha(t-s)}\|\int_s^t V_r dX_r\| ds+\int_{t-\delta}^t e^{-\alpha(t-s)}\|\int_s^t V_r dX_r\| ds\right)\notag\\
\leq &\frac{2}{\alpha}e^{-\alpha\delta}\sup_{t\in[0,T]}\|\int_0^t V_r dX_r\| +\sup_{t\in[\delta,T]}\int_{t-\delta}^t e^{-\alpha(t-s)}\|\int_s^t V_r dX_r\| ds.\notag
\end{align}
Let $N=\max\{k\in\mathbb{Z}:k\delta<T\}$.  Then $P$-a.s.
\begin{align}\label{second_term_bound}
&\sup_{t\in[\delta,T]}\int_{t-\delta}^t e^{-\alpha(t-s)}\|\int_s^t V_r dX_r\| ds\\
\leq& \max_{k=0,...,N-1}\sup_{t\in [(k+1)\delta, (k+2)\delta]}\int_{k\delta}^{t}e^{-\alpha(t-s)}\|\int_s^tV_rdX_r\|ds\notag\\
=&  \max_{k=0,...,N-1}\sup_{t\in [(k+1)\delta, (k+2)\delta]}\int_{k\delta}^{t}e^{-\alpha(t-s)}\|\int_{k\delta}^tV_rdX_r-\int_{k\delta}^sV_rdX_r\|ds\notag\\
\leq &\frac{2}{ \alpha} \max_{k=0,...,N-1}\sup_{t\in [k\delta, (k+2)\delta]}\|\int_{k\delta}^tV_rdX_r\|.\notag
\end{align}
Combining \req{first_term_bound} and \req{second_term_bound} and using the inequality 
\begin{align}
\sup_{t\in[0,\delta]}\|\int_0^t V_s dX_s\| \leq  \max_{k=0,...,N-1}\sup_{t\in [k\delta, (k+2)\delta]}\|\int_{k\delta}^tV_rdX_r\|
\end{align}
 gives the $P$-a.s. bound
\begin{align}
&\sup_{t\in[0,T]}\|\Phi(t)\int_0^t\Phi^{-1}(s) V_s dX_s\|\leq \sup_{t\in[0,\delta]}\|\int_0^t V_s dX_s\|+e^{-\alpha \delta}\sup_{t\in[\delta,T]}\|\int_0^t V_s dX_s\|\notag\\
&+\frac{2}{ \alpha}\sup_{s\in[0,T]}\|B(s)\|\left(\sup_{t\in[0,\delta]}\|\int_0^t V_r dX_r\|+e^{-\alpha\delta}\sup_{t\in[0,T]}\|\int_0^t V_r dX_r\| \right.\notag\\
&\left.+ \max_{k=0,...,N-1}\sup_{t\in [k\delta, (k+2)\delta]}\|\int_{k\delta}^tV_rdX_r\|\right)\\
\leq& (1+\frac{4}{ \alpha}\sup_{s\in[0,T]}\|B(s)\|)\left(e^{-\alpha\delta}\sup_{t\in[0,T]}\|\int_0^t V_r dX_r\|\right.\notag\\
&\left. +\max_{k=0,...,N-1}\sup_{t\in [k\delta, (k+2)\delta]}\|\int_{k\delta}^tV_rdX_r\|\right)\notag
\end{align}
as claimed.
\end{proof}

It will also be useful to recall the following (see \cite{karatzas2014brownian}).
\begin{lemma}\label{martingale_lemma}
 If $M\in\mathcal{M}^{c,loc}$, $V\in L^2_{loc}(M)$, and $E[\int_{t_0}^t V^2d[M]_s]<\infty$ for all $t$ then $\int_{t_0}^tV_s dM_s$ is a martingale.
\end{lemma}

\subsection{  Limit of the Momentum Process }
 In this section we show three propositions about convergence of the momentum process $p_t^m=mv_t^m$ to zero as $m\rightarrow 0$.

\begin{proposition}\label{E_sup_p_bound}
 For any $p> 0$, $T>0$, and $0<\beta<p/2$ we have
\begin{align}
E[\sup_{t\in[0,T]} \|p_t^m\|^{p}]=O(m^{\beta})\text{ as $m\rightarrow 0$}.
\end{align}
\end{proposition}
\begin{proof}
The strategy here is to first rewrite the equation for $p^m_t$ so that the stochastic integral term has the same form as the left hand side of \req{Phi_int_equality}.  Using the bound \req{Phi_int_bound2} we will then be able to show that both terms decay as $m\rightarrow 0$.  The first term will decay exponentially, and the second term will decay because the stocastic integrals will be taken over ``small" time intervals.

The momentum solves the SDE
\begin{align}\label{p_SDE}
dp^m_t=&(F(u^m_t)-\frac{1}{m}\gamma(u^m_t)p^m_t)dt+\sigma(u^m_t) dW_t.
\end{align}
This is a linear SDE on $\mathbb{R}^n$ where  $F(u^m_t)$, $-\frac{1}{m}\gamma(u^m_t)$, and $\sigma(u^m_t)$ are pathwise continuous adapted vector or matrix-valued processes, and so its unique solution can be written in terms of $u^m_s$
\begin{align}\label{p_SDE_sol}
p^m_t=\Phi(t) \left(p^m_0+\int_0^t \Phi^{-1}(s) F(u^m_s)ds+\int_0^t\Phi^{-1}(s)\sigma(u^m_s) dW_s\right)
\end{align}
where $\Phi(t)$ is the adapted $C^1$ process that pathwise solves the IVP
\begin{align}
\dot{\Phi}(t)=-\frac{1}{m}\gamma(u^m_t)\Phi(t),\hspace{2mm}\Phi(0)=I.
\end{align}
This technique of utilizing the explicit solution of a linear SDE to obtain estimates is used in \cite{Freidlin2004}, where the case of constant, scalar drag on flat Euclidean space is studied.

By Assumption \ref{drag_assump}, the symmetric part of $-\frac{1}{m}\gamma(u)$  has eigenvalues bounded above by $-\gamma_1/m<0$ with the bound uniform in $u$.  Therefore, using \req{fundamental_matrix}, for $s\leq t$,
\begin{align}
\|\Phi(t)\Phi^{-1}(s)\|\leq e^{-\gamma_1(t-s)/m} 
\end{align}
and hence for every $T>0$, $p\geq 1$,
\begin{align}
\sup_{t\in[0,T]}\|p^m_t\|^p\leq&\sup_{t\in[0,T]}3^{p-1}\left(e^{-\gamma_1p t /m} m^p\|v_0\|^p+\left(\int_0^t  e^{-\gamma_1(t-s)/m}\| F(u^m_s)\|ds\right)^p\right.\notag\\
&\left.+\|\Phi(t)\int_0^t \Phi^{-1}(s)\sigma(u^m_s) dW_s\|^p\right)\label{p_bound1}\\
\leq &3^{p-1}\left(m^p\|v_0\|^p+ \frac{m^p}{\gamma_1^p}\|F\|_\infty^p+\sup_{t\in[0,T]}\|\Phi(t)\int_0^t \Phi^{-1}(s)\sigma(u^m_s) dW_s\|^p\right),\notag
\end{align}
where $\|F\|_\infty$ denotes the supremum of $\|F(u)\|$ over $u$ and we have employed the inequality 
\begin{align}\label{tight_power_bound}
\left(\sum_{i=1}^N a_i\right)^p\leq N^{p-1}\sum_{i=1}^N a_i^p
\end{align}
for every $p\geq 1$, $N\in\mathbb{N}$.

Taking the $p$th power of \req{Phi_int_bound2}, for any $\delta$ with  $0<\delta<T$ we have the $P$-a.s. bound

\begin{align}\label{Phi_int_bound3}
&\sup_{t\in[0,T]}\|\Phi(t)\int_0^t\Phi^{-1}(s)\sigma(u^m_s) dW_s\|^p\notag\\
\leq& 2^{p-1}(1+\frac{4}{ \gamma_1}\sup_{s\in[0,T]}\|\gamma(u^m_t)\|)^p\bigg(e^{-p\delta\gamma_1/m}\sup_{t\in[0,T]}\|\int_0^t \sigma(u^m_r) dW_r\|^p \\
&+\max_{k=0,...,N-1}\sup_{t\in [k\delta, (k+2)\delta]}\|\int_{k\delta}^t  \sigma(u^m_r) dW_r\|^p\bigg)\notag
\end{align}
where $N=\max\{k\in\mathbb{Z}:k\delta<T\}$.

We now return to bounding the momentum using \req{p_bound1}.  As was done in \req{p_bound1}, the supremum of a quantity $\|A(u)\|$ will be denoted by $\|A\|_\infty$ for an arbitrary matrix or vector-valued function $A$ (rather than by the more precise but less readable $\| \|A\| \|_\infty$).
\begin{align}
&\sup_{t\in[0,T]}\|p^m_t\|^p\notag\\
\leq &3^{p-1}\bigg[m^p\|v_0\|^p+ \frac{m^p}{\gamma_1^p}\|F\|_\infty^p+2^{p-1}(1+\frac{4}{ \gamma_1}\|\gamma\|_\infty)^p\bigg(e^{-p\delta\gamma_1/m}\sup_{t\in[0,T]}\|\int_0^t \sigma(u^m_r) dW_r\|^p \notag\\
&+\left(\sum_{k=0}^{N-1}\sup_{t\in [k\delta, (k+2)\delta]}\|\int_{k\delta}^t  \sigma(u^m_r) dW_r\|^{pq}\right)^{1/q}\bigg)\bigg]
\end{align}
where we used \req{Phi_int_bound3} and the fact that the sup norm on $\mathbb{R}^N$ is bounded by the $\ell^q$ norm for any $q\geq 1$. We will take $q>1$.

Taking the expected value and then using H\"older's inequality on the expectations we get
\begin{align}
&E[\sup_{t\in [0,T]}\|p^m_t\|^p] \leq 3^{p-1}\bigg[m^p\|v_0\|^p+ \frac{m^p}{\gamma_1^p}\|F\|_\infty^p\notag\\
&+2^{p-1}(1+\frac{4}{ \gamma_1}\|\gamma\|_\infty)^p\bigg(e^{-p\delta\gamma_1/m}E[\sup_{t\in[0,T]}\|\int_0^t \sigma(u^m_r) dW_r\|^{pq}]^{1/q} \\
&+\left(\sum_{k=0}^{N-1}E[\sup_{t\in [k\delta, (k+2)\delta]}\|\int_{k\delta}^t  \sigma(u^m_r) dW_r\|^{pq})]\right)^{1/q}\bigg)\bigg].\notag
\end{align}

The Burkholder-Davis-Gundy inequalities (see for example Theorem 3.28 in \cite{karatzas2014brownian}), for $d>1$ imply the existence of a constant $C_{d,n}>0$ such that
\begin{align}
 &E[\sup_{0\leq s\leq T}\|\int_{0 }^s \sigma(u^m_r) dW_r\|^{d}]\leq C_{d,n} E[(\int_{0 }^T \|\sigma(u^m_r)\|_F^2 dr)^{d/2}]
\end{align}
where $\|\cdot\|_F$ denotes the Frobenius (or Hilbert-Schmidt) norm.

Therefore,  letting $\delta=m^{1-\kappa}$ for $0<\kappa<1$,   we find
\begin{align}
&E[\sup_{t\in [0,T]}\|p^m_t\|^p] \leq 3^{p-1}\bigg[m^p\|v_0\|^p+ \frac{m^p}{\gamma_1^p}\|F\|_\infty^p\notag\\
&+2^{p-1} C_{pq,n}^{1/q}(1+\frac{4}{ \gamma_1}\|\gamma\|_\infty)^p\bigg(e^{-p\gamma_1/m^\kappa} E[(\int_{0 }^{T} \|\sigma(u^m_r)\|_F^2 dr)^{pq/2}]^{1/q} \notag\\
&+\left( \sum_{k=0}^{N-1} E[(\int_{k\delta }^{(k+2)\delta} \|\sigma(u^m_r)\|_F^2 dr)^{pq/2}]\right)^{1/q}\bigg)\bigg]\\
\leq &3^{p-1}\bigg[m^p\|v_0\|^p+ \frac{m^p}{\gamma_1^p}\|F\|_\infty^p+2^{p-1} C_{pq,n}^{1/q}\|\sigma\|_{F,\infty}^{p}(1+\frac{4}{ \gamma_1}\|\gamma\|_\infty)^p\bigg(e^{-p\gamma_1/m^\kappa} T^{p/2}\notag\\
&+2^{p/2}\left( N \delta^{pq/2} \right)^{1/q}\bigg)\bigg],\notag
\end{align}
where we define $\|\sigma\|_{F,\infty}=\sup_u \|\sigma(u)\|_F$.

$N\delta<T$, hence
\begin{align}
N\delta^{pq/2}<T\delta^{pq/2-1}=Tm^{(1-\kappa)(pq/2-1)}.
\end{align}
Therefore
\begin{align}
E[\sup_{t\in[0,T]}\|p^m_t\|^p]=O(m^{(1-\kappa)(p/2-1/q)}).
\end{align}
For any $0<\beta<p/2$ we can choose $0<\kappa<1$ and $q>1$ so that $(1-\kappa)(p/2-1/q)=\beta$, thereby proving the claim for $p\geq1$.

For any $0<p<1$ and $0<\beta<p/2$, take $q\geq 1$.  Then  $\beta q/p<q/2$ so, using H\"older's inequality, we find
\begin{align}
E[\sup_{t\in[0,T]}\|p^m_t\|^p]\leq E[(\sup_{t\in[0,T]}\|p^m_t\|^p)^{q/p}]^{p/q}= O(m^\beta).
\end{align}
\end{proof}

If we don't take the supremum over $t$ inside the expectation, we can prove a stronger decay result.
\begin{proposition}\label{sup_E_p_bound}
 For any $q>0$ and any $m_0>0$ there exists a $C>0$ such that  
\begin{align}
\sup_{ t\in[0, \infty)}E[\|p_t^m\|^{q}]\leq Cm^{q/2}
\end{align}
 for all $0<m\leq m_0$.
\end{proposition}
\begin{proof}
Let $\gamma_1 >\alpha>0$ and define the process $z^m_t=e^{\alpha t/m}p^m_t$.  By It\^o's formula
\begin{align}\label{z_eq}
z^m_t=&p^m_0+\int_0^t \frac{\alpha}{m} z^m_sds+\int_0^te^{\alpha s/m}(F(u^m_s)-\frac{1}{m}\gamma(u^m_s)p^m_s)ds+\int_0^te^{\alpha s/m}\sigma(u^m_s)dW_s\\
=&p^m_0+\int_0^t[ -\frac{1}{m}(\gamma(u^m_s)-\alpha )z^m_s+e^{\alpha s/m}F(u^m_s)]ds+\int_0^te^{\alpha s/m}\sigma(u^m_s)dW_s.
\end{align}
This holds for any $\alpha$, but we will need $0<\alpha<\gamma_1$ later.

Now applying the It\^o formula again to $ \|z^m_t\|^2=(z^m_t)^Tz^m_t$ we find
\begin{align}
\|z^m_t\|^2=&\|p^m_0\|^2+2\left(\int_0^t (z^m_s)^T(-\frac{1}{m}(\gamma(u^m_s)-\alpha )z^m_s+e^{\alpha s/m}F(u^m_s))ds\right.\notag\\
&\left.+\int_0^te^{\alpha s/m}(z^m_s)^T\sigma(u^m_s)dW_s \right)\\
&+\sum_i[\sum_j\int_0^\cdot e^{\alpha s/m}\sigma^i_j(u^m_s)dW^j_s,\sum_j\int_0^\cdot e^{\alpha s/m}\sigma^i_j(u^m_s)dW^j_s]_t.\notag
\end{align}
The quadratic variation term is
\begin{align}
&\sum_i[\sum_j\int_0^\cdot e^{\alpha s/m}\sigma^i_j(u^m_s)dW^j_s,\sum_j\int_0^\cdot e^{\alpha s/m}\sigma^i_j(u^m_s)dW^j_s]_t\notag\\
=&\int_0^te^{2\alpha s/m}\|\sigma(u^m_s)\|_F^2ds.
\end{align}
  Therefore
\begin{align}
\|z^m_t\|^2=&\|p^m_0\|^2-\frac{2}{m}\int_0^t (z^m_s)^T(\gamma(u^m_s)-\alpha )z^m_sds+2\int_0^te^{\alpha s/m}(z^m_s)^T F(u^m_s)ds\notag\\
&+\int_0^te^{2\alpha s/m}\|\sigma(u^m_s)\|_F^2ds+2\int_0^te^{\alpha s/m}(z^m_s)^T\sigma(u^m_s)dW_s .
\end{align}

First we will show the result for $q=2p$ with $p$ a positive integer.  Using It\^o's formula one more time,  we obtain
\begin{align}
\|z^m_t\|^{2p}=&\|p^m_0\|^{2p}-\frac{2}{m}\int_0^t p\|z^m_s\|^{2(p-1)}(z^m_s)^T(\gamma(u^m_s)-\alpha )z^m_sds\notag\\
&+2\int_0^t p\|z^m_s\|^{2(p-1)}e^{\alpha s/m}(z^m_s)^T F(u^m_s)ds\notag\\
& +\int_0^t p\|z^m_s\|^{2(p-1)}e^{2\alpha s/m}\|\sigma(u^m_s)\|_F^2ds\\
&+\frac{p(p-1)}{2}\int_{t_0}^t\|z^m_s\|^{2(p-2)}4e^{2\alpha s/m}\|\sigma^T(u^m_s)z_s^m\|^2ds\notag\\
&+2\int_0^t p\|z^m_s\|^{2(p-1)}e^{\alpha s/m}(z^m_s)^T\sigma(u^m_s)dW_s.\notag
\end{align}

By Assumption \ref{drag_assump}, we have $y^T\gamma(u)y\geq \gamma_1 \|y\|^2$ for all $u\in F_O(M)$, $y\in\mathbb{R}^n$. Hence, defining $\epsilon=\gamma_1-\alpha>0$, we deduce from here the following upper bound on the norm of $p_t^m$.
\begin{align}\label{y_bound}
\|p^m_t\|^{2p}\leq&e^{-2p\alpha t/m}\|p^m_0\|^{2p}-\frac{2p\epsilon}{m}\int_0^t e^{-2p\alpha (t-s)/m}\|p^m_s\|^{2p}ds\notag\\
&+2p\int_0^t e^{-2p\alpha(t- s)/m}\|p^m_s\|^{2(p-1)}(p^m_s)^T F(u^m_s)ds\notag\\
& +p\int_0^t e^{-2p\alpha(t- s)/m}\|p^m_s\|^{2(p-1)}\|\sigma(u^m_s)\|_F^2ds\\
&+2p(p-1)\int_{t_0}^te^{-2p\alpha (t-s)/m}\|p^m_s\|^{2(p-2)}\|\sigma^T(u^m_s)p_s^m\|^2ds\notag\\
&+2p\int_0^t e^{-2p\alpha(t- s)/m}\|p^m_s\|^{2(p-1)}(p^m_s)^T\sigma(u^m_s)dW_s.\notag
\end{align}
Using Lemma \ref{martingale_lemma}, the following computation shows that 
\begin{equation}
M_t\equiv\int_0^t e^{2p\alpha s/m}\|p^m_s\|^{2(p-1)}(p^m_s)^T\sigma(u^m_s)dW_s
\end{equation}
 is a martingale.
\begin{align}\label{martingale_proof}
&E[\int_0^t \|e^{2p\alpha s/m}\|p^m_s\|^{2(p-1)}(p^m_s)^T\sigma(u^m_s)\|^2 ds]\\
\leq& e^{4p\alpha t/m}\|\sigma\|^2_{\infty} t E[\sup_{s\in[0,t]} \|p^m_s\|^{4p-2}]<\infty,\notag
\end{align}
where in the last step, the expectation is finite by Proposition \ref{E_sup_p_bound}.
Using this and taking the expected value of \req{y_bound} we find, for any $a>0$,
\begin{align}
E[\|p^m_t\|^{2p}]\leq&e^{-2p\alpha t/m}\|p^m_0\|^{2p}-\frac{2p\epsilon}{m}\int_0^t e^{-2p\alpha (t-s)/m}E[\|p^m_s\|^{2p}]ds\\
&+2p\int_0^t e^{-2p\alpha(t- s)/m}E[\|p^m_s\|^{2(p-1)}(p^m_s)^T F(u^m_s)]ds\notag\\
& +p\int_0^t e^{-2p\alpha(t- s)/m}E[\|p^m_s\|^{2(p-1)}\|\sigma(u^m_s)\|_F^2]ds\notag\\
&+2p(p-1)\int_{t_0}^te^{-2p\alpha (t-s)/m}E[\|p^m_s\|^{2(p-2)}\|\sigma^T(u^m_s)p_s^m\|^2]ds\notag\\
\leq &e^{-2p\alpha t/m}\|p^m_0\|^{2p}-\frac{2p\epsilon}{m}\int_0^t e^{-2p\alpha (t-s)/m}E[\|p^m_s\|^{2p}]ds\\
&+p\int_0^t e^{-2p\alpha(t- s)/m}E[\|p^m_s\|^{2(p-1)}(a^2\|p^m_s\|^2 +\|F\|_\infty^2/a^2)]ds\notag\\
& +(p\|\sigma\|_{F,\infty}^2+2p(p-1)\|\sigma\|_\infty^2)\int_0^t e^{-2p\alpha(t- s)/m}E[\|p^m_s\|^{2(p-1)}]ds\notag\\
= &e^{-2p\alpha t/m}\|p^m_0\|^{2p}-p\left(\frac{2\epsilon}{m}-a^2\right)\int_0^t e^{-2p\alpha (t-s)/m}E[\|p^m_s\|^{2p}]ds\notag\\
&+p\|F\|_\infty^2/a^2\int_0^t e^{-2p\alpha(t- s)/m}E[\|p^m_s\|^{2(p-1)}]ds\\
& +(p\|\sigma\|_{F,\infty}^2+2p(p-1)\|\sigma\|_\infty^2)\int_0^t e^{-2p\alpha(t- s)/m}E[\|p^m_s\|^{2(p-1)}]ds.\notag
\end{align}
Leting $a^2=2\epsilon/m$ yields
\begin{align}
E[\|p^m_t\|^{2p}]\leq&e^{-2p\alpha t/m}\|p^m_0\|^{2p} +(\frac{pm}{2\epsilon}\|F\|_\infty^2+p\|\sigma\|_{F,\infty}^2+2p(p-1)\|\sigma\|_\infty^2)\notag\\
&\times \int_0^t e^{-2p\alpha(t- s)/m}E[\|p^m_s\|^{2(p-1)}]ds.
\end{align}
For $p=1$ this becomes
\begin{align}
E[\|p^m_t\|^2]\leq &m^2e^{-2\alpha t/m}\|v_0\|^2+\frac{m^2}{4\epsilon\alpha}\|F\|_\infty^2+\frac{m}{2\alpha }\|\sigma\|_{F,\infty}^2.
\end{align}
Hence
\begin{align}
\sup_{t\in[0,\infty)}E[\|p^m_t\|^2]\leq m^2\|v_0\|^2+\frac{m^2}{4\epsilon\alpha}\|F\|_\infty^2+\frac{m}{2\alpha }\|\sigma\|_{F,\infty}^2\leq Cm
\end{align}
for all $0<m\leq m_0$.  So the claim holds for $p=1$.

Now suppose it holds for $p-1$.  Then
\begin{align}
&\sup_{ t\in[0, \infty)}E[\|p^m_t\|^{2p}]\notag\\
\leq&m^{2p}\|v_0\|^{2p} +(\frac{pm}{2\epsilon}\|F\|_\infty^2+p\|\sigma\|_{F,\infty}^2+2p(p-1)\|\sigma\|_\infty^2)\\
&\times \sup_{t\in[0,\infty)} \int_0^t e^{-2p\alpha(t- s)/m}E[\|p^m_s\|^{2(p-1)}]ds\notag\\
\leq&m^{2p}\|v_0\|^{2p} +(\frac{pm}{2\epsilon}\|F\|_\infty^2+p\|\sigma\|_{F,\infty}^2+2p(p-1)\|\sigma\|_\infty^2)\\
&\times \sup_{t\in[0,\infty)} \int_0^t e^{-2p\alpha(t- s)/m}ds\sup_{t\in[0,\infty)}E[\|p^m_t\|^{2(p-1)}]\notag\\
\leq&m^{2p}\|v_0\|^{2p} +Cm^{p-1}\frac{m}{2p\alpha}(\frac{pm}{2\epsilon}\|F\|_\infty^2+p\|\sigma\|_{F,\infty}^2+2p(p-1)\|\sigma\|_\infty^2)\notag\\
\leq &\tilde C m^p
\end{align}
for all $0<m\leq m_0$. This proves the claim for all positive integer $p$ by induction.

Finally, let $q>0$ be arbitrary.  Take $p\in\mathbb{N}$ with $2p> q$.   Using H\"older's inequality we find
\begin{align}
\sup_{ t\in[0, \infty)}E[\|p^m_t\|^{q}]\leq \left(\sup_{ t\in[0, \infty)}E[\|p^m_t\|^{2p}]\right)^{q/(2p)}\leq (Cm^p)^{q/(2p)}=\tilde C m^{q/2}
\end{align}
for all $0<m\leq m_0$.  Therefore the result holds for any $q>0$.
\end{proof}

We can use the above decay results to prove that certain integrals with respect to the momentum process also vanish in the limit.

\begin{proposition}\label{quad_p_int_bound}
For any smooth function $f$ on $F_O(M)$, any $T> 0$, $p> 0$, and any $\alpha,\beta=1,...,n$ we have
\begin{align}
E[\sup_{0\leq t\leq T}|\int_{0}^t f(u^m_s) d((p^m_s)^\alpha (p^m_s)^\beta)|^{p}]=O(m^{p/2}) \text{ as } m\rightarrow 0,
\end{align}
where $\alpha$, $\beta$ refer to the components of the momentum process in the standard basis for $\mathbb{R}^n$.
\end{proposition}
\begin{proof}
First assume $p>1$. Integrating by parts, and using the fact that $u^m$ is pathwise $C^1$, we get
\begin{align}
\int_{0}^t f(u^m_s) d((p^m_s)^\alpha (p^m_s)^\beta)=&f(u^m_t) (p^m_t)^\alpha (p^m_t)^\beta-f(u^m_{0}) (p^m_{0})^\alpha (p^m_{0})^\beta\notag\\
&-\int_{0}^t (p^m_s)^\alpha (p^m_s)^\beta \frac{d}{ds}f(u^m_s) ds.
\end{align}
From the original SDE, \req{SDE_u_m2}, $u^m$ satisfies the ODE $\dot{u}=H_v(u)$, so
\begin{align}
\frac{d}{ds}f(u^m_s)= \frac{1}{m}H_{p^m_s}(u^m_s)[f].
\end{align}
$F_O(M)$ is compact, hence $f$ is bounded.  Therefore, decomposing the vector $p_s^m$ in the standard basis $e_\nu$ of $\mathbb{R}^n$ and using the notation $H_\nu$ introduced in Lemma \ref{horizontal_vf}, we have
\begin{align}
&E[\sup_{0\leq t\leq T}|\int_{0}^t f(u^m_s) d((p^m_s)^\alpha (p^m_s)^\beta)|^{p}]\leq 3^{p-1}\left(E[ \sup_{0\leq t\leq T} |f(u^m_t) (p^m_t)^\alpha (p^m_t)^\beta|^{p}]\right.\notag\\
&\left.+E[|f(u^m_{0}) (p^m_{0})^\alpha (p^m_{0})^\beta|^{p}]+E[\sup_{0\leq t\leq T}|\int_{0}^t (p^m_s)^\alpha (p^m_s)^\beta  \frac{1}{m}H_{p^m_s}(u^m_s)[f]ds|^{p}]\right)\notag\\
\leq&3^{p-1}\left(\|f\|_\infty^{p} E[ \sup_{0\leq t\leq T} \|p^m_t\|^{2p}]+m^{2p} \|f\|_\infty^{p}\|v_{0}\|^{2p}\right)\\
&+\frac{3^{p-1}}{m^{p}}E[\sup_{0\leq t\leq T}|\int_{0}^t (p^m_s)^\alpha (p^m_s)^\beta ( p^m_s)^\nu H_{\nu}(u^m_s)[f]ds|^{p}].\notag
\end{align}
We have assumed $p>1$, so  we can use H\"older's inequality with  exponents $p$ and $p/(p-1)$ to estimate the last term.  Using boundedness of $H_\nu[f]$ (implied by compactness of $F_O(M)$), this gives
\begin{align}
&E[\sup_{0\leq t\leq T}|\int_{0}^t (p^m_s)^\alpha (p^m_s)^\beta ( p^m_s)^\nu H_{\nu}(u^m_s)[f]ds|^{p}]\notag\\
\leq&E[\left(\int_{0}^T |(p^m_s)^\alpha (p^m_s)^\beta ( p^m_s)^\nu H_{\nu}(u^m_s)[f]|ds\right)^{p}]\notag\\
\leq&T^{p-1}E[\int_{0}^T |(p^m_s)^\alpha (p^m_s)^\beta ( p^m_s)^\nu H_{\nu}(u^m_s)[f]|^{p}ds]\\
\leq&CT^{p-1}E[\int_{0}^T \|p^m_s\|^{3p} ds]=CT^{p-1}\int_{0}^T E[\|p^m_s\|^{3p}] ds\notag\\
\leq &CT^{p} \sup_{0\leq s\leq T}E[\|p^m_s\|^{3p}]\notag
\end{align}
for some $C>0$.

By Proposition \ref{E_sup_p_bound}, for any $0<\kappa<p$ we obtain
\begin{align}
&E[\sup_{0\leq t\leq T}|\int_{0}^t f(u^m_s) d((p^m_s)^\alpha (p^m_s)^\beta)|^{p}]\\
\leq& O(m^\kappa)+\frac{3^{p-1}CT^{p}}{m^{p}} \sup_{0\leq s\leq T}E[\|p^m_s\|^{3p}].\notag
\end{align}
Applying Proposition \ref{sup_E_p_bound} to the second term we get
\begin{align}
&E[\sup_{0\leq t\leq T}|\int_{0}^t f(u^m_s) d((p^m_s)^\alpha (p^m_s)^\beta)|^{p}]\notag\\
\leq& O(m^\kappa)+\frac{3^{p-1}CT^{p}}{m^{p}}\tilde C m^{3p/2}\\
\leq& O(m^\kappa)+O( m^{p/2}).\notag
\end{align}
Taking $\kappa=p/2$ gives the result for all $p>1$.  The result for arbitrary $p>0$ then follows by an application of H\"older's inequality, as in Propositions \ref{E_sup_p_bound} and \ref{sup_E_p_bound}.

\end{proof}

\section{Calculation of the Limiting SDE}\label{sec:limit_SDE}
We now manipulate the SDE,  \reqr{SDE_u_m2}{SDE_v_m2}, for $(u_t^m,v_t^m)$  to extract the terms that survive when $m\rightarrow 0$ in order to derive a candidate for the limiting SDE.  The actual convergence proof will be given in Section \ref{convergence_proof_sec}.

 In order to express the equations in terms of Lebesgue and  It\^o integrals on some $\mathbb{R}^l$, we consider the composition of a function $f\in C^\infty(F_O(M))$ with $u^m_t$.  The following equations are then satisfied on $\mathbb{R}$ and $\mathbb{R}^n$ respectively,
\begin{align}
&df(u^m_t)=H_{v^m_t}(u^m_t)[f]dt,\\
&dv^m_t=\frac{1}{m}(F(u^m_t)-\gamma(u^m_t)v^m_t)dt+\frac{1}{m}\sigma(u^m_t) dW_t.\label{v_SDE_sec_6}
\end{align}

We know that the momentum, $p^m_t$, converges to zero (in various senses) and our objective will be to separate out such terms.  We begin by solving for  $v^m_t dt$ in the equation for $dv^m_t$.  Note that, by assumption, $\gamma^s(u)$ is positive definite for all $u$, and hence $\gamma(u)$ is invertible for every $u$. Therefore
\begin{align}\label{v_dt}
v^m_t dt=-\gamma^{-1}(u^m_t)dp^m_t+\gamma^{-1}(u^m_t)F(u^m_t)dt+\gamma^{-1}(u^m_t)\sigma(u^m_t) dW_t.
\end{align}
Using the linearity of $H_v$ in $v$, the first equation can be written
\begin{align}\label{du_1}
df(u^m_t)=&(v^m_t)^\nu H_\nu (u^m_t)[f]dt,
\end{align}
where we are  again using the notation $H_\nu$ introduced in Lemma \ref{horizontal_vf} and components of the $\mathbb{R}^n$-valued process $v^m_t$ will always refer to the standard basis.  For any $n\times l$ matrix $A$, we will let $H_A(u)$ denote the element of $(\mathbb{R}^l)^*$ with action $w\mapsto H_{Aw}(u)$. For any $\mathbb{R}^l$-valued semimartingale, $X$,  will also write $H_A(u)dX$ as shorthand for the contraction $H_{Ae_\eta}(u)dX^\eta$.

  With these notations, after substituting \req{v_dt} into \req{du_1} we obtain
\begin{align}\label{df_eq2}
df(u^m_t)=&-H_\nu (u^m_t)[f](\gamma^{-1})^\nu_\mu(u^m_t)d(p^m_t)^\mu+H_\nu (u^m_t)[f](\gamma^{-1})_\mu^\nu(u^m_t)F^\mu(u^m_t)dt\notag\\
&+H_\nu (u^m_t)[f](\gamma^{-1})_\mu^\nu(u^m_t)\sigma^\mu_\eta(u^m_t) dW^\eta_t\\
=&-H_{\gamma^{-1}(u^m_t)} (u^m_t)[f]d(p^m_t)+H_{(\gamma^{-1} F) (u^m_t)} (u^m_t)[f]dt\notag\\
&+H_{(\gamma^{-1}\sigma)(u^m_t)} (u^m_t)[f] dW_t.\notag
\end{align}

\begin{remark}
Here and in the following, components of matrix or $\mathbb{R}^n$-valued functions on the frame bundle, such as $\gamma(u)$ or $F(u)$,  as well as the implied sums over repeated indices, will always refer to the standard basis.  We emphasize that these components are unrelated to local coordinates on $M$ or $F_O(M)$ and do not in any way imply that the statements have a local character.  Rather, by lifting a tensor or vector field from $M$ to a  matrix or vector valued function on $F_O(M)$, we are able to speak about its components in each frame in a globally defined manner.
\end{remark}

The second and third terms in \req{df_eq2} are independent of the momentum, so we focus on the first term.  Integrating by parts and using the fact that $u_t^m$ is pathwise $C^1$ and satisfies the ODE \req{SDE_u_m2} we get
\begin{align}
&H_{\gamma^{-1}(u^m_t)} (u^m_t)[f]d(p^m_t)\notag\\
=&d(H_{\gamma^{-1}(u^m_t)p^m_t} (u^m_t)[f])-(p^m_t)^\nu d(H_{\gamma^{-1}(u^m_t)e_\nu} (u^m_t)[f])\notag\\
=&d(H_{\gamma^{-1}(u^m_t)p^m_t} (u^m_t)[f])-(p^m_t)^\nu \frac{d}{dt}(H_{\gamma^{-1}(u^m_t)e_\nu} (u^m_t)[f])dt\\
=&d(H_{\gamma^{-1}(u^m_t)p^m_t} (u^m_t)[f])-(p^m_t)^\nu H_{v^m_t}(u^m_t)[(\gamma^{-1})_\nu^\mu H_{\mu}[f]]dt\notag\\
=&d(H_{\gamma^{-1}(u^m_t)p^m_t} (u^m_t)[f])-m(v^m_t)^\nu (v^m_t)^\xi  K^f_{\nu\xi}(u^m_t)dt\notag
\end{align}
by \req{du_1}, where we define 
\begin{align}\label{Kf_def}
K^f_{\nu\xi}(u)=H_{\xi}(u)[(\gamma^{-1})_\nu^\mu] H_{\mu}(u)[f]+(\gamma^{-1})_\nu^\mu(u) H_{\xi}(u)[H_{\mu}[f]].
\end{align}
The equation for $f(u^m_t)$ then becomes
\begin{align}
df(u^m_t)=&-d(H_{\gamma^{-1}(u^m_t)p^m_t} (u^m_t)[f])+m(v^m_t)^\nu (v^m_t)^\mu K^f_{\nu\mu}(u^m_t)dt\\
&+H_{(\gamma^{-1} F) (u^m_t)} (u^m_t)[f]dt+H_{(\gamma^{-1}\sigma)(u^m_t)} (u^m_t)[f] dW_t.\notag
\end{align}
To simplify $m(v^m_t)^\nu (v^m_t)^\mu dt$  we follow \cite{Hottovy2014} and compute
\begin{align}
&d(m(v^m_t)^\nu m(v^m_t)^\mu)\\
=&m(v^m_t)^\nu d(m(v^m_t)^\mu)+m(v^m_t)^\mu d(m(v^m_t)^\nu) +d([m(v^m_t)^\nu,m(v^m_t)^\mu])\notag\\
=&m(v^m_t)^\nu d(m(v^m_t)^\mu)+m(v^m_t)^\mu d(m(v^m_t)^\nu) +\sum_\delta\sigma^\nu_\delta(u^m_t) \sigma^\mu_\delta(u^m_t) dt,\notag
\end{align}
where we have used  the SDE for $(v^m_t)^\nu$,  \req{v_SDE_sec_6}.   Defining 
\begin{align}
\Sigma^{\nu\mu}(u)=\sum_\delta\sigma^\nu_\delta(u) \sigma^\mu_\delta(u)
\end{align}
 and using the SDE for $(v^m_t)^\nu$ again, we get
\begin{align}
&d(m(v^m_t)^\nu m(v^m_t)^\mu)\notag\\
=&m(v^m_t)^\nu (F^\mu(u^m_t)-\gamma^\mu_\xi(u^m_t)(v^m_t)^\xi)dt+m(v^m_t)^\nu\sigma^\mu_\eta(u^m_t) dW^\eta_t\\
&+m(v^m_t)^\mu(F^\nu(u^m_t)-\gamma^\nu_\xi(u^m_t)(v^m_t)^\xi)dt+m(v^m_t)^\mu\sigma^\nu_\eta(u^m_t) dW^\eta_t +\Sigma^{\nu\mu}(u^m_t) dt.\notag
\end{align}
Therefore
\begin{align}\label{KE_gamma_eq}
&m(v^m_t)^\nu \gamma^\mu_\xi(u^m_t)(v^m_t)^\xi dt+m(v^m_t)^\mu\gamma^\nu_\xi(u^m_t)(v^m_t)^\xi dt\notag\\
=&- d(m(v^m_t)^\nu m(v^m_t)^\mu)+\left((p^m_t)^\nu F^\mu(u^m_t)+(p^m_t)^\mu F^\nu(u^m_t)\right)dt\\
&+\left((p^m_t)^\nu\sigma^\mu_\eta(u^m_t)+(p^m_t)^\mu\sigma^\nu_\eta(u^m_t)  \right)dW^\eta_t+\Sigma^{\nu\mu}(u^m_t) dt.\notag
\end{align}

For scalar $\gamma$, one can immediately solve for  $m(v^m_t)^\nu (v^m_t)^\mu dt$.  To handle the general case, we utilize the technique developed in \cite{Hottovy2014}.  First we rewrite \req{KE_gamma_eq} in integral notation:
\begin{align}
&\int_0^t \left(m(v^m_s)^\nu \gamma^\mu_\xi(u^m_s)(v^m_s)^\xi+m(v^m_s)^\mu\gamma^\nu_\xi(u^m_s)(v^m_t)^\xi \right)ds\\
=&- m(v^m_t)^\nu m(v^m_t)^\mu+m(v_0)^\nu m(v_0)^\mu+\int_0^t\left((p^m_s)^\nu F^\mu(u^m_s)+(p^m_s)^\mu F^\nu(u^m_s)\right)ds\notag\\
&+\int_0^t\left((p^m_s)^\nu\sigma^\mu_\eta(u^m_s) +(p^m_s)^\mu\sigma^\nu_\eta(u^m_s) \right)dW^\eta_s+\int_0^t\Sigma^{\nu\mu}(u^m_s) ds.\notag
\end{align}

The formula for the left hand side implies that the right hand side, which we denote by $C_t^{\mu\nu}$, is $C^1$  $P$-a.s.  and so we can differentiate both sides with respect to $t$ to obtain
\begin{align}
m(v^m_t)^\nu \gamma^\mu_\xi(u^m_t)(v^m_t)^\xi+m(v^m_t)^\mu\gamma^\nu_\xi(u^m_t)(v^m_t)^\xi =\dot{C}_t^{\mu\nu}.
\end{align}
Define the matrix $V_t^{\nu\mu}=m(v^m_t)^\nu(v^m_t)^\mu$.  The above equation can be written in matrix form
\begin{align}
\gamma V+V\gamma^T=\dot{C}.
\end{align}
This is a Lyapunov equation, where $-\gamma$ has eigenvalues with real part bounded above by $-\gamma_1$.  Hence, we can solve uniquely for $V$,
\begin{align}
V=\int_0^\infty e^{-y\gamma}\dot{C} e^{-y\gamma^T}dy,
\end{align}
 see for example Theorem 6.4.2 in \cite{ortega2013matrix}.  Integrating with respect to $t$, we get  
\begin{align}
\int_0^t m(v^m_s)^\nu(v^m_s)^\mu ds=\int_0^t \int_0^\infty (e^{-y\gamma(u^m_s)})^{\nu}_\eta  (e^{-y\gamma(u^m_s)})^{\mu}_\xi dy (\dot{C}_s)^{\eta\xi}ds.
\end{align}
Define the functions $G^{\nu\mu}_{\eta\xi}(u)=\int_0^\infty (e^{-y \gamma(u)})^\nu_\eta (e^{-y\gamma(u)})^\mu_\xi dy$.  Using the dominated convergence theorem, along with the formula for the derivative of the matrix exponential found in \cite{exp_deriv}, one can prove that $G$ is a smooth function of $u$.  Therefore, the $G^{\nu\mu}_{\eta\xi}(u^m_t)$ are  semimartingales and
\begin{align}\label{mvv_eq}
&\int_0^t m(v^m_s)^\beta(v^m_s)^\alpha ds=\int_0^t  G^{\beta\alpha}_{\mu\nu}(u^m_s)dC_s^{\mu\nu}=-\int_0^t  G^{\beta\alpha}_{\mu\nu}(u^m_s)d((p^m_s)^\nu (p^m_s)^\mu)\notag\\
&+\int_0^t G^{\beta\alpha}_{\mu\nu}(u^m_s)\left((p^m_s)^\nu F^\mu(u^m_s)+(p^m_s)^\mu F^\nu(u^m_s)\right)ds\\
&+\int_0^t G^{\beta\alpha}_{\mu\nu}(u^m_s)\left((p^m_s)^\mu\sigma^\nu_\eta(u^m_s) +(p^m_s)^\nu\sigma^\mu_\eta(u^m_s) \right)dW^\eta_s\notag\\
&+\int_0^t G^{\beta\alpha}_{\mu\nu}(u^m_s)\Sigma^{\mu\nu}(u^m_s) ds.\notag
\end{align}
Note that if we define 
\begin{equation}\label{J_def}
J^{\beta\alpha}(u)= G^{\beta\alpha}_{\mu\nu}(u)\Sigma^{\mu\nu}(u)
\end{equation}
 then $J$ is symmetric and is the unique solution to the Lyapunov equation
\begin{align}\label{J_Lyap_eq}
\gamma J+J\gamma^T=\Sigma.
\end{align}

Using \req{mvv_eq} we find that $f(u^m_t)$  satisfies 
\begin{align}\label{dfu_m}
&df(u^m_t)\notag\\
=&H_{(\gamma^{-1} F) (u^m_t)} (u^m_t)[f]dt+J^{\beta\alpha}(u^m_t) K^f_{\beta\alpha}(u^m_t) dt+H_{(\gamma^{-1}\sigma)(u^m_t)} (u^m_t)[f] dW_t\notag\\
&-d(H_{\gamma^{-1}(u^m_t)p^m_t} (u^m_t)[f])- K^f_{\beta\alpha}(u^m_t)G^{\beta\alpha}_{\mu\nu}(u^m_t)d((p^m_t)^\nu (p^m_t)^\mu)\\
&+K^f_{\beta\alpha}(u^m_t)G^{\beta\alpha}_{\mu\nu}(u^m_t)\left((p^m_t)^\nu F^\mu(u^m_t)+(p^m_t)^\mu F^\nu(u^m_t)\right)dt\notag\\
&+K^f_{\beta\alpha}(u^m_t)G^{\beta\alpha}_{\mu\nu}(u^m_t)\left((p^m_t)^\mu\sigma^\nu_\eta(u^m_t) +(p^m_t)^\nu\sigma^\mu_\eta(u^m_t) \right)dW^\eta_t.\notag
\end{align}

Based on the results of Section \ref{sec:p_zero}, we expect that terms involving $p^m_t$ will converge to zero as $m\rightarrow 0$.  Therefore, if it exists, we would expect the limiting process $u_t$ to satisfy 
\begin{align}\label{df_eq_Ito}
df(u_t)=&H_{(\gamma^{-1} F) (u_t)} (u_t)[f]dt+J^{\beta\alpha}(u_t) K^f_{\beta\alpha}(u_t) dt+H_{(\gamma^{-1}\sigma)(u_t)} (u_t)[f] dW_t
\end{align}
for every $f\in C^\infty(F_O(M))$. Note that by Lemma \ref{horiz_lift_vf_formula}, the first term is the horizontal lift of the vector field $(\gamma^{-1}F)(x)$ on $M$ to $F_O(M)$, 
\begin{align}
H_{(\gamma^{-1} F)(u)} (u)=(\gamma^{-1}F)^h(u).
\end{align}

In order to express \req{df_eq_Ito} as an SDE on the manifold $F_O(M)$ we need to rewrite it using the Stratonovich integral. This is accomplished by the following lemma.
\begin{lemma}
An $F_O(M)$-valued semimartingale $u_t$ satisfies the SDE
\begin{align}\label{limiting_SDE}
du_t=&(\gamma^{-1}F)^h(u_t)dt+H_{(\gamma^{-1}\sigma)(u_t)} (u_t)\circ dW_t\\
&-\frac{1}{2}\sum_\alpha (\gamma^{-1}(u_t))^\eta_\mu\sigma_\alpha^\mu(u_t)(\gamma^{-1}(u_t))^\xi_\nu H_{ \eta} (u_t) [ \sigma_\alpha^\nu] H_{ \xi} (u_t)dt\notag \\
&-\frac{1}{2}\left((\gamma^{-1}(u_t))^\eta_\mu J^{\mu\chi}(u_t)\gamma^\nu_\chi(u_t)-J^{\eta\nu}(u_t)\right)  H_{ \eta} (u_t) [(\gamma^{-1})^\xi_\nu] H_{ \xi}(u_t) dt\notag\\
&-\frac{1}{2}J^{\xi\nu}(u_t)(\gamma^{-1}(u_t))^\eta_\nu  [H_{ \eta} ,H_{ \xi}](u_t) dt\notag
\end{align}
iff
\begin{align}\label{limiting_SDE_Ito}
df(u_t)=&H_{(\gamma^{-1} F) (u_t)} (u_t)[f]dt+J^{\beta\alpha}(u_t) K^f_{\beta\alpha}(u_t) dt+H_{(\gamma^{-1}\sigma)(u_t)} (u_t)[f] dW_t
\end{align}
for every $f\in C^\infty(F_O(M))$.  The sum over $\alpha$ in \req{limiting_SDE} is taken over the standard basis for $\mathbb{R}^k$.
\end{lemma}
\begin{proof}
By Lemma \ref{Ito_SDE_lemma}, $u_t$ satisfies the SDE \req{limiting_SDE} iff
\begin{align}\label{limiting_SDE_Ito2}
df(u_t)=&H_{(\gamma^{-1}F)(u_t)}(u_t)[f]dt+H_{(\gamma^{-1}\sigma)(u_t)} (u_t)[f] dW_t\\
&+\frac{1}{2}(\gamma^{-1}(u_t))^\eta_\mu \sigma_\alpha^\mu(u_t)  H_{ \eta} (u_t) [(\gamma^{-1})^\xi_\nu \sigma_\beta^\nu H_{ \xi}[f] ] d[W^\alpha,W^\beta]_t\notag\\
&-\frac{1}{2}(\gamma^{-1}(u_t))^\eta_\mu\sigma_\alpha^\mu(u_t)(\gamma^{-1}(u_t))^\xi_\nu H_{ \eta} (u_t) [ \sigma_\beta^\nu] H_{ \xi}[f] (u_t)\delta^{\alpha\beta}dt \notag\\
&-\frac{1}{2}\left((\gamma^{-1}(u_t))^\eta_\mu J^{\mu\chi}(u_t)\gamma^\nu_\chi(u_t)-J^{\eta\nu}(u_t)\right)  H_{ \eta} (u_t) [(\gamma^{-1})^\xi_\nu] H_{ \xi}(u_t)[f] dt\notag\\
&-\frac{1}{2}J^{\xi\nu}(u_t)(\gamma^{-1}(u_t))^\eta_\nu  [H_{ \eta} ,H_{ \xi}](u_t) [f]dt\notag
\end{align}
for all $f\in C^\infty(F_O(M))$.

The  Lyapunov equation, \req{J_Lyap_eq}, implies $\gamma^{-1}\Sigma(\gamma^{-1})^T=\gamma^{-1}J+J(\gamma^{-1})^T$. Together with the symmetry of $J$ and $\Sigma$, this yields the following for any $F_O(M)$-valued semimartingale $u_t$ and any $f\in C^\infty(F_O(M))$.
\begin{align}\label{limit_SDE_derivation}
&\frac{1}{2}(\gamma^{-1}(u_t))^\eta_\mu \sigma_\alpha^\mu(u_t)  H_{ \eta} (u_t) [(\gamma^{-1})^\xi_\nu \sigma_\beta^\nu H_{ \xi}[f] ] d[W^\alpha,W^\beta]_t\\
=&\frac{1}{2}(\gamma^{-1}(u_t))^\eta_\mu \sigma_\alpha^\mu(u_t)  \left(H_{ \eta} (u_t) [(\gamma^{-1})^\xi_\nu] \sigma_\beta^\nu(u_t) H_{ \xi}(u_t)[f] \right.\notag\\
&\left.+(\gamma^{-1}(u_t))^\xi_\nu H_{ \eta} (u_t) [ \sigma_\beta^\nu] H_{ \xi}[f] (u_t) +(\gamma^{-1}(u_t))^\xi_\nu\sigma_\beta^\nu(u_t) H_{ \eta} (u_t) [ H_{ \xi}[f]]\right)\delta^{\alpha\beta}dt\notag\\
=&\frac{1}{2}(\gamma^{-1}(u_t))^\eta_\mu \Sigma^{\mu\nu}(u_t)  H_{ \eta} (u_t) [(\gamma^{-1})^\xi_\nu] H_{ \xi}(u_t)[f] dt\\
&+\frac{1}{2}(\gamma^{-1}(u_t))^\eta_\mu\sigma_\alpha^\mu(u_t)(\gamma^{-1}(u_t))^\xi_\nu H_{ \eta} (u_t) [ \sigma_\beta^\nu] H_{ \xi}[f] (u_t)\delta^{\alpha\beta}dt\notag \\
&+\frac{1}{4}(\gamma^{-1}(u_t))^\eta_\mu \Sigma^{\mu\nu}(u_t)(\gamma^{-1}(u_t))^\xi_\nu\left( H_{ \eta} (u_t) [ H_{ \xi}[f]]+ H_{ \xi} (u_t) [ H_{ \eta}[f]]\right)dt\notag\\
=&\frac{1}{2}(\gamma^{-1}(u_t))^\eta_\mu \Sigma^{\mu\nu}(u_t)  H_{ \eta} (u_t) [(\gamma^{-1})^\xi_\nu] H_{ \xi}(u_t)[f] dt\\
&+\frac{1}{2}(\gamma^{-1}(u_t))^\eta_\mu\sigma_\alpha^\mu(u_t)(\gamma^{-1}(u_t))^\xi_\nu H_{ \eta} (u_t) [ \sigma_\beta^\nu] H_{ \xi}[f] (u_t)\delta^{\alpha\beta}dt\notag \\
&+\frac{1}{2}J^{\xi\nu}(u_t)(\gamma^{-1}(u_t))^\eta_\nu \left( [H_{ \eta} ,H_{ \xi}](u_t) [f]+ 2H_{ \xi} (u_t) [ H_{ \eta}[f]]\right)dt\notag\\
=& J^{\xi\nu}(u_t)K^f_{\nu\xi}(u_t)dt+\frac{1}{2}(\gamma^{-1}(u_t))^\eta_\mu\sigma_\alpha^\mu(u_t)(\gamma^{-1}(u_t))^\xi_\nu H_{ \eta} (u_t) [ \sigma_\beta^\nu] H_{ \xi}[f] (u_t)\delta^{\alpha\beta}dt \notag\\
&+\frac{1}{2}\left((\gamma^{-1}(u_t))^\eta_\mu \Sigma^{\mu\nu}(u_t)- 2J^{\eta\nu}(u_t)\right)  H_{ \eta} (u_t) [(\gamma^{-1})^\xi_\nu] H_{ \xi}(u_t)[f] dt\\
&+\frac{1}{2}J^{\xi\nu}(u_t)(\gamma^{-1}(u_t))^\eta_\nu  [H_{ \eta} ,H_{ \xi}](u_t) [f]dt.\notag
\end{align}

Using the Lyapunov equation, \req{J_Lyap_eq}, one more time gives 
\begin{equation}
\gamma^{-1}\Sigma-2J=\gamma^{-1}J\gamma^T-J.
\end{equation}
 Therefore, combining the result of \req{limit_SDE_derivation} with \req{limiting_SDE_Ito2}  yields \req{limiting_SDE_Ito}.
\end{proof}

The proposed limiting SDE, \req{limiting_SDE}, includes a drift term generated by the vector field $S(u)=S^h(u)+S^v(u)$, where
\begin{align}
S^h(u)=&-\frac{1}{2}\sum_\alpha (\gamma^{-1}(u_t))^\eta_\mu\sigma_\alpha^\mu(u_t)(\gamma^{-1}(u_t))^\xi_\nu H_{ \eta} (u_t) [ \sigma_\alpha^\nu] H_{ \xi} (u_t)\label{S_h} \\
&-\frac{1}{2}\left((\gamma^{-1}(u_t))^\eta_\mu J^{\mu\chi}(u_t)\gamma^\nu_\chi(u_t)-J^{\eta\nu}(u_t)\right)  H_{ \eta} (u_t) [(\gamma^{-1})^\xi_\nu] H_{ \xi}(u_t), \notag\\
S^v(u)=&-\frac{1}{2}J^{\xi\nu}(u_t)(\gamma^{-1}(u_t))^\eta_\nu  [H_{ \eta} ,H_{ \xi}](u_t).\label{S_v}
\end{align} 
This is on top of the horizontal lift of the deterministic force to the frame bundle, $(\gamma^{-1}F)^h(u)$.  $S^h(u)$ is a linear combination of the $H_\alpha(u)$'s, hence is a horizontal vector field on $F_O(M)$. It corresponds to an additional ``force'' on the particle's position, which we call  the {\em noise-induced drift}. A nonzero noise-induced drift requires either a non-trivial state dependence of the noise coefficients, $\sigma^\mu_\nu$, or  non-trivial state dependence of the drag,  together with $\gamma^{-1}J\gamma^T\neq J$, but it does not require any deterministic forcing, $F$, to be present in the original system.  An analogous phenomenon was derived in Euclidean space of arbitrary dimension in \cite{Hottovy2014}. See \cite{GiovanniReview} for a recent review of  noise induced drift in systems with multiplicative noise.

 For a torsion free connection on $M$, such as the Levi-Civita connection employed here, the $[H_{ \eta} ,H_{ \xi}]$ are vertical vector fields on $F_O(M)$  that can be expressed in terms of the curvature tensor \cite{hsu2002stochastic}.  Therefore the $S^v$ is a vertical vector field on $F_O(M)$  that results in an additional rotational ``force" on the particle's frame.

Note that if $\gamma$ is a scalar then $\gamma$ and $J$ commute, hence the second term of $S^h$ vanishes.  $S^v$ also vanishes in this case, due to the symmetry of $J^{\xi\eta}$ combined with the antisymmetry of  $[H_{ \eta} ,H_{ \xi}]$.

\subsection{Some Special Cases}
Before we prove convergence to the proposed limiting equation, \req{limiting_SDE}, we will study its form   in several  cases of interest. For all these cases, we make the assumption that $k=n$ and $\sigma^\mu_\nu$ comes from a $\binom{1}{1}$-tensor field, $\sigma(x)$, on $M$,  as in \req{sigma_tensor}. The following lemma will allow us to obtain a geometric formula for the noise induced drift in this case. A similar study could be done of the case of \req{k_vf_case}, but we don't pursue this here. 
\begin{lemma}
Let $\sigma,\tau,\kappa$,  be  $\binom{1}{1}$-tensor fields on $M$ and define the matrix-valued functions $\sigma(u),\tau(u),\kappa(u)$ as done in \req{sigma_tensor}, i.e. $\sigma(u)=u^{-1}\sigma(\pi(u))u$ and similarly for $\tau$, $\kappa$.  Let $Y$ be the vector field on $M$ defined by
\begin{align}
Y^\chi(x)=g^{\nu\rho}(x)\tau^\beta_\nu(x)\kappa_\delta^\chi(x) (\nabla_\beta\sigma_\rho^\delta)(x).
\end{align}
The horizontal lift of $Y$ is given by
\begin{align}
Y^h(u)=\sum_\eta \tau^\nu_\eta(u)H_\nu(u)[\sigma^\xi_\eta] \kappa_\xi^\chi H_\chi(u).
\end{align}
\end{lemma}
\begin{proof}
We work in the domain of a coordinate chart and local o.n. frame $E_\alpha$  and use the  notation  defined in Section \ref{sec:ode_in_coords}.  Let $\sigma^\mu_\nu(x)$ be the components of $\sigma$ in the frame $E_\alpha$.  Therefore
\begin{align}
 \sigma_\xi^\nu(u)=(h^{-1})_\mu^\nu\sigma^\mu_\alpha(x)   h^\alpha_\xi
\end{align}
and similarly for $\tau$, $\kappa$. Using Section \ref{sec:ode_in_coords} and the fact that $h\in O(\mathbb{R}^n)$, we obtain
\begin{align}
&\sum_\eta \tau^\nu_\eta(u)H_\nu(u)[\sigma^\xi_\eta]=\sum_\eta \tau^\nu_\eta(u) h^\beta_\nu( (\Lambda^{-1})_\beta^i(x)\partial_i\sigma^\xi_\eta(u) - h^\kappa_\mu A_{\beta \kappa}^\delta(x)\partial_{h^\delta_\mu}\sigma^\xi_\eta(u))\\
=&\sum_\eta \tau^\beta_\nu(x)   h^\nu_\eta( (\Lambda^{-1})_\beta^i(x)\partial_i((h^{-1})_\chi^\xi\sigma^\chi_\rho(x)   h^\rho_\eta) - h^\kappa_\mu A_{\beta \kappa}^\delta(x)\partial_{h^\delta_\mu}((h^{-1})_\chi^\xi\sigma^\chi_\rho(x)   h^\rho_\eta))\notag\\
=& \tau^\beta_\nu(x)   (\sum_\eta h^\rho_\eta h^\nu_\eta (\Lambda^{-1})_\beta^i(x)(h^{-1})_\chi^\xi\partial_i(\sigma^\chi_\rho )(x)   -\sum_\eta h^\nu_\eta h^\kappa_\mu A_{\beta \kappa}^\delta(x)\partial_{h^\delta_\mu}((h^{-1})_\chi^\xi)\sigma^\chi_\rho(x)   h^\rho_\eta\notag\\
& -\sum_\eta h^\nu_\eta h^\kappa_\mu A_{\beta \kappa}^\delta(x)(h^{-1})_\chi^\xi\sigma^\chi_\rho(x)   \partial_{h^\delta_\mu}(h^\rho_\eta))\\
=& \tau^\beta_\nu(x)   (\delta^{\rho\nu} (\Lambda^{-1})_\beta^i(x)(h^{-1})_\chi^\xi\partial_i(\sigma^\chi_\rho )(x)   -\delta^{\rho\nu}h^\kappa_\mu A_{\beta \kappa}^\delta(x)((-h^{-1} \partial_{h^\delta_\mu}(h) h^{-1})_\chi^\xi)\sigma^\chi_\rho(x) \notag \\
& -\sum_\eta h^\nu_\eta h^\kappa_\mu A_{\beta \kappa}^\delta(x)(h^{-1})_\chi^\xi\sigma^\chi_\rho(x)   \delta^{\rho}_\delta\delta^{\mu}_\eta)\\
=& \tau^\beta_\nu(x)  \delta^{\nu\rho}  ((\Lambda^{-1})_\beta^i(x)(h^{-1})_\chi^\xi\partial_i(\sigma^\chi_\rho )(x)   +h^\kappa_\mu A_{\beta \kappa}^\delta(x)(h^{-1})^\xi_\phi \partial_{h^\delta_\mu}(h^\phi_\psi)( h^{-1})_\chi^\psi\sigma^\chi_\rho(x) \notag \\
& -A_{\beta \rho}^\delta(x)(h^{-1})_\chi^\xi\sigma^\chi_\delta(x)  )\\
=& \tau^\beta_\nu(x)  \delta^{\nu\rho}  ((\Lambda^{-1})_\beta^i(x)(h^{-1})_\chi^\xi\partial_i(\sigma^\chi_\rho )(x)   +h^\kappa_\mu A_{\beta \kappa}^\delta(x)(h^{-1})^\xi_\delta ( h^{-1})_\chi^\mu\sigma^\chi_\rho(x) \notag \\
& -A_{\beta \rho}^\delta(x)(h^{-1})_\chi^\xi\sigma^\chi_\delta(x)  )\\
=& \tau^\beta_\nu(x)  \delta^{\nu\rho}((\Lambda^{-1})_\beta^i(x)\partial_i(\sigma^\chi_\rho )(x)   + A_{\beta \kappa}^\chi(x)\sigma^\kappa_\rho(x)  -A_{\beta \rho}^\delta(x)\sigma^\chi_\delta(x)  ) (h^{-1})_\chi^\xi.
\end{align} 

On the last line, the terms in parentheses are the components of the covariant derivative of $\sigma$ in the direction $E_\beta$, therefore
\begin{align}
&\sum_\eta \tau^\nu_\eta(u)H_\nu(u)[\sigma^\xi_\eta]= \tau^\beta_\nu(x)  \delta^{\nu\rho}(\nabla_\beta\sigma)_\rho^\chi(x) (h^{-1})_\chi^\xi.
\end{align} 
All tensor components are taken in the o.n. frame $E_\alpha$, hence $g^{\mu\nu}=\delta^{\mu\nu}$.  Therefore
\begin{align}
&\sum_\eta \tau^\nu_\eta(u)H_\nu(u)[\sigma^\xi_\eta]\kappa_\xi^\chi(u) e_\chi=u^{-1}(g^{\nu\rho}\tau^\beta_\nu(x)  (\nabla_\beta\sigma)_\rho^\alpha(x)\kappa^\mu_\alpha(x) E_\mu) =u^{-1}Y(x).
\end{align} 
This proves 
\begin{align}
\sum_\eta \tau^\nu_\eta(u)H_\nu(u)[\sigma^\xi_\eta] \kappa_\xi^\chi(u) H_\chi(u)=H_{u^{-1}Y(\pi(u))}=Y^h(u).
\end{align}
\end{proof}

Several applications of the above lemma can be used to show that, for tensor $\sigma$, the noise induced drift is a horizontal lift.
\begin{corollary}
Define the smooth tensor fields on $M$, 
\begin{align}
\Sigma^{\xi\eta}(x)=&\sigma^\xi_\alpha(x)g^{\alpha\beta}(x)\sigma_\beta^\eta(x),\\
 J^{\mu\nu}(x)= &\int_0^\infty (e^{-y \gamma(x)})^\mu_\delta\Sigma^{\delta\eta}(x)  (e^{-y\gamma(x)})^\nu_\eta   dy, \\
L^{\mu\nu}(x)=& \int_0^\infty (\gamma^{-1}(x)e^{-y\gamma(x)})^\mu_\alpha  \Sigma^{\alpha\beta}(x) (\gamma(x)e^{-y \gamma(x)})^\nu_\beta dy.
\end{align}
Then the noise induced drift vector field, $S^h$, is the horizontal lift of the vector field, $Y$, on $M$ defined by
\begin{align}
 Y^\chi=-\frac{1}{2}g^{\nu\rho}(\gamma^{-1}\sigma)^\beta_\nu(\gamma^{-1})_\delta^\chi \nabla_\beta\sigma_\rho^\delta-\frac{1}{2}(L^{\beta\rho}-J^{\beta\rho})\nabla_\beta(\gamma^{-1})_\rho^\chi.
\end{align}

\end{corollary}

An important special case is when the drag and noise satisfy the fluctuation-dissipation relation.
\begin{corollary}[A Particle Satisfying the Fluctuation-Dissipation Relation]
Suppose the  fluctuation-dissipation relation is satisfied,
\begin{align}
\gamma^\mu_\nu(x)=\frac{1}{2k_BT}\sigma^\mu_\alpha(x) g^{\alpha\beta}(x)\sigma_{\beta}^\delta(x) g_{\delta\nu}(x),
\end{align}
where $T$ is the temperature and $k_B$ is the Boltzmann constant. Then $\gamma^\mu_\nu(u)=\frac{1}{2k_BT}\Sigma^{\mu\alpha}(u)\delta_{\alpha\nu}$, $J^{\mu\nu}=k_BT\delta^{\mu\nu}$, $S^v=0$, and $S^h$ is the horizontal lift of
\begin{align}
 Y^\chi=-k_BT g^{\beta\mu}(\sigma^{-1})_\mu ^\rho(\gamma^{-1})_\delta^\chi \nabla_\beta\sigma_\rho^\delta.
\end{align}
\end{corollary}

Next we specialize to scalar drag.
\begin{corollary}[A Particle with Scalar Drag]
If $\gamma^\mu_\nu(x)=\gamma(x)\delta^\mu_\nu$ for some $\gamma\in C^\infty(M)$ then $S^v=0$ and
\begin{align}
S^h=-\frac{1}{2}(\gamma^{-2}Y)^h,\hspace{2mm} Y^\chi=g^{\nu\rho}\sigma^\beta_\nu \nabla_\beta\sigma_\rho^\chi.
\end{align} 
 The proposed limiting SDE is then
\begin{align}
du_t=&(\gamma^{-1}F)^h(u_t)dt-\frac{1}{2}(\gamma^{-2}Y)^h(u_t)dt+\gamma^{-1}(\pi(u_t))H_{\sigma(u_t)} (u_t) \circ dW_t.
\end{align}
\end{corollary}

Next, we consider the case where both $\gamma$ and  $\sigma$  are scalars.
\begin{corollary}[A Particle with Scalar Drag and Noise]
Specializing further to $\gamma^\mu_\nu(x)=\gamma(x)\delta^\mu_\nu$ and $\sigma^\mu_\nu(x)=\sigma(x) \delta^\mu_\nu$ for some $\gamma,\sigma\in C^\infty(M)$ we obtain $S^v=0$,
\begin{align}
S^h=&-\frac{1}{2} (\gamma^{-2}\sigma \nabla \sigma)^h,
\end{align}
and hence the proposed limiting SDE is
\begin{align}
du_t=&(\gamma^{-1}F)^h(u_t)dt-\frac{1}{2} (\gamma^{-2}\sigma \nabla \sigma)^h(u_t)dt+(\gamma^{-1}\sigma)(\pi(u_t))H (u_t) \circ dW_t.
\end{align}
In particular, if $\sigma$ is a constant then the noise induced drift vanishes.
\end{corollary}

Finally, we arrive at the case that leads to Brownian motion in the limit.
\begin{corollary}[Brownian Motion]
If $\gamma=\sigma$ are constant scalars and $F=0$ then the proposed limiting SDE is
\begin{align}
du_t=&H (u_t) \circ dW_t,
\end{align}
whose solution is the lift of a Brownian motion on $M$ to the frame bundle \cite{hsu2002stochastic}.
\end{corollary}
Once we prove convergence in the next section, this last corollary will complete the objective set forth in the introduction, namely deriving Brownian motion on the manifold as the small mass limit of a noisy inertial system with drag.\\

\subsection{Behavior Under Change of Frame}
The transformation properties of the proposed  limiting equation \req{limiting_SDE} under right multiplication by an orthogonal matrix are given by the following lemma.
\begin{lemma}
Let $h\in O(\mathbb{R}^n)$ and $u_t$ be the solution to \req{limiting_SDE} corresponding to the initial condition $u_0$.  Suppose $\sigma$ is of the form \req{sigma_tensor}. Then
\begin{align}
\tilde u_t=u_t h
\end{align}
is the solution to \req{limiting_SDE}  corresponding to the initial condition $u_0 h$ with the Wiener process $W_t$ replaced by the Wiener process $\tilde W_t=h^{-1} W_t$.
\end{lemma}
\begin{proof}
$\tilde u_t$ is a semimartingale with initial condition $u_0 h$.  Right multiplication by $h$ is a diffeomorphism, and hence by Lemma \ref{SDE_diffeo_lemma}, $\tilde u_t$ is a solution to
\begin{align}\label{u_tilde_eq}
d\tilde u_t=&\left((R_{h})_*(\gamma^{-1}F)^{h}+(R_{h})_*S^h+(R_{h})_*S^v\right)(\tilde u_t)dt+(R_{h})_*( H_{\gamma^{-1}\sigma}) (\tilde u_t) \circ dW_t.
\end{align}
By Lemma \ref{horiz_lift_vf_formula}, horizontal lifts are invariant under right translation. The vertical term is also right invariant by the following computation:
\begin{align}
&((R_{h})_*S^v)(\tilde u_t)=(R_{h})_*(S^v( u_t))\notag\\
=&-\frac{1}{2}J^{\xi\nu}(u_t)(\gamma^{-1}(u_t))^\eta_\nu  [(R_{h})_*H_{ \eta} ,(R_{h})_*H_{ \xi}](\tilde u_t)\\
=&-\frac{1}{2} (h^{-1})_\xi^\beta J^{\xi\nu}(u_t)(\gamma^{-1}(u_t))^\eta_\nu (h^{-1})_\eta^\alpha [H_{ \alpha} ,H_{ \beta}](\tilde u_t).\notag
\end{align}
Using the definition of $\gamma(u)$ along with \req{J_def} and  one can show that
\begin{align}
(h^{-1})_\xi^\beta J^{\xi\nu}(u_t)(\gamma^{-1}(u_t))^\eta_\nu (h^{-1})_\eta^\alpha=J^{\beta\nu}(\tilde u_t)(\gamma^{-1}(\tilde u_t))^\alpha_\nu 
\end{align}
and hence
\begin{align}
((R_{h})_*S^v)(\tilde u_t)=S^v(\tilde u_t).
\end{align}

The last term in \req{u_tilde_eq} is
\begin{align}
&(R_{h})_*( H_{\gamma^{-1}\sigma}) (\tilde u_t)=H_{h^{-1}(\gamma^{-1}\sigma)(u_t)}(\tilde u_t)=H_{(\gamma^{-1}\sigma)(\tilde u_t)h^{-1}}(\tilde u_t).
\end{align}
 Hence we have
\begin{align}
d\tilde u_t=&(\gamma^{-1}F)^{h}(\tilde u_t)+S^h(\tilde u_t)+S^v(\tilde u_t)+( H_{\gamma^{-1}\sigma}) (\tilde u_t) \circ d\tilde W_t
\end{align}
as claimed.
\end{proof}

We end this section by showing that, if one is only interested in the statistics of the position process, then the vertical drift term, $S^v$, can be neglected.
\begin{proposition}
Let $u_t$ be the solution to the proposed limiting SDE, \req{limiting_SDE}. Suppose $\sigma$ is of the form \req{sigma_tensor}.  Then there exists an $O(\mathbb{R}^n)$-valued semimartingale, $h_t$, with $h_0=I$, the identity matrix, such that $\tilde u_t=u_th_t$ is a solution to the  limiting SDE, minus the vertical component of the drift, i.e. a solution to
\begin{align}\label{limiting_SDE_horiz}
d\tilde u_t=&(\gamma^{-1}F)^h(\tilde u_t)dt+H_{(\gamma^{-1}\sigma)(\tilde u_t)} (\tilde u_t)\circ d\tilde W_t\\
&-\frac{1}{2}\sum_\alpha (\gamma^{-1}(\tilde u_t))^\eta_\mu\sigma_\alpha^\mu(\tilde u_t)(\gamma^{-1}(\tilde u_t))^\xi_\nu H_{ \eta} (\tilde u_t) [ \sigma_\alpha^\nu] H_{ \xi} (\tilde u_t)dt \notag\\
&-\frac{1}{2}\left((\gamma^{-1}(\tilde u_t))^\eta_\mu J^{\mu\chi}(\tilde u_t)\gamma^\nu_\chi(\tilde u_t)-J^{\eta\nu}(\tilde u_t)\right)  H_{ \eta} (\tilde u_t) [(\gamma^{-1})^\xi_\nu] H_{ \xi}(\tilde u_t) dt,\notag
\end{align}
where the Wiener process $\tilde W_t=\int_0^t h^{-1}_s dW_s$ is used in the Stratonovich integral.

 In particular, since $u_th_t$ and $u_t$ have the same position process and uniqueness in law holds for a SDE on a compact manifold (as can be seen by employing a smooth embedding in some $\mathbb{R}^l$ along with the corresponding result in \cite{karatzas2014brownian}), the distribution of the position process is unchanged by the vertical component of the drift, even if one doesn't make a change to the Wiener process $\tilde W_t$.
\end{proposition}
\begin{proof}
Note that for any $O(\mathbb{R}^n)$-valued semimartingale, $h_t$, the process  $\tilde W_t=\int_0^t h^{-1}_s dW_s$ is a continuous $\mathbb{R}^n$-valued local martingale with quadratic covariation
\begin{align}
[\tilde W^\alpha,\tilde W^\beta]_t=\int_0^t (h^{-1}_s)^\alpha_\delta (h^{-1}_s)^\beta_\eta \delta^{\delta\eta}ds =t\delta^{\alpha\beta}.
\end{align}
Hence,  $\tilde W_t$ is a Wiener process by Levy's theorem \cite{karatzas2014brownian}.

We will show more generally that, given a vector field, $Y(x)$, on $M$, a $\binom{1}{1}$-tensor field, $\tau(x)$, on $M$, and a vertical vector field, $V(u)$, on $F_O(M)$  (not to be confused with the forcing, \req{assump_V_form}) one can go from a solution of the SDE
\begin{align}\label{genericSDE_with_vert}
du_t=&(Y)^h(u_t)dt+V(u_t)dt+H_{\tau(u_t)} (u_t)\circ d W_t
\end{align}
to a solution of
\begin{align}\label{genericSDE_no_vert}
d\tilde u_t=&(Y)^h(\tilde u_t)dt+H_{\tau(\tilde u_t)} (\tilde u_t)\circ d\tilde W_t
\end{align}
with the same initial condition in the manner described above.

Let $\phi:F_O(M)\times O(\mathbb{R}^n)\rightarrow F_O(M)$ be the right action and, for $u\in F_O(M)$, define $\phi_u(h)=\phi(u,h)$.  These are both smooth maps.  For each $A\in \mathfrak{o}(\mathbb{R}^n)$, the Lie algebra of $O(\mathbb{R}^n)$,  we obtain a smooth vertical vector field $V_A$ on $F_O(M)$ defined by
\begin{align}
V_A(u)=(\phi_u)_*A.
\end{align}
If we let $A_i$ be a basis for  $\mathfrak{o}(\mathbb{R}^n)$ then $V_i(u)\equiv V_{A_i}(u)$ form a basis for the vertical subspace at each $u$ \cite{hsu2002stochastic}.  Therefore we can write $V(u)=V^i(u)V_i(u)$ for some smooth functions $V^i$ on $F_O(M)$.

Let $X_i$ be the smooth right invariant vector fields on $O(\mathbb{R}^n)$ defined by $X_i(h)=(R_h)_* A_i$. Consider the SDE on the compact manifold $F_O(M)\times O(\mathbb{R}^n)$,
\begin{align}
du_t=&(Y)^h(u_t)dt+V(u_t)dt+H_{\tau(u_t)} (u_t)\circ d W_t,\\
dh_t=& - V^i(u_t)X_i(h)dt.
\end{align}
By compactness, a unique solution, $(u_t,h_t)$, corresponding to the initial condition $(u_0,I)$ exists for all $t\geq 0$. Note that the first component, $u_t$, is a solution to \req{genericSDE_with_vert} with the initial condition $u_0$.

Next we see how the vector fields of this SDE behave under the pushforward. 
\begin{align}\label{vector_push1}
&\phi_* ((Y)^h(u)+V(u), -V^i(u)X_i(h))\notag\\
=&(R_h)_*((Y)^h(u)+V(u))-(\phi_u)_*(V^i(u)X_i(h))\notag\\
=&(Y)^h(uh)+(R_h)_*(V(u))-V^i(u)(\phi_u\circ R_h)_* A_i\notag\\
=&(Y)^h(uh)+(R_h)_* V(u)-V^i(u) (R_h)_* V_i(u)\notag\\
=&(Y)^h(uh)
\end{align}
and
\begin{align}\label{vector_push2}
\phi_*(H_{\tau(u)e_\alpha} (u),0)=(R_h)_*H_{\tau(u)e_\alpha} (u)=H_{h^{-1}\tau(u)e_\alpha} (uh)=(h^{-1})^\beta_\alpha H_{\tau(uh)e_\beta} (uh).
\end{align}

Using \req{vector_push1} and \req{vector_push2}, we can derive the SDE satisfied by the semimartingale $\tilde u_t=u_th_t=\phi(u_t,h_t)$. For any $f\in C^\infty(F_O(M))$,  $\tilde f=f\circ\phi$ is a smooth function on $F_O(M)\times O(\mathbb{R}^n)$, hence 
\begin{align}
f(\tilde u_t)=&\tilde f(u_0,I)+\int_0^t((Y)^h(u_s)+V(u_s), -V^i(u_s)X_i(h_s))[\tilde f]ds\notag\\
&+\int_0^t H_{\tau(u_s)} (u_s)[\tilde f]\circ d W_s\\
=&f(u_0)+\int_0^t\phi_*((Y)^h(u_s)+V(u_s), -V^i(u_s)X_i(h_s))[ f]ds\notag\\
&+\int_0^t \phi_*H_{\tau(u_s)} (u_s)[ f]\circ d W_s\\
=&f( u_0)+\int_0^t(Y)^h(\tilde u_s)[ f]ds+\int_0^t  H_{\tau(\tilde u_s)e_\beta} (\tilde u_s)[ f](h_s^{-1})^\beta_\alpha\circ d W^\alpha_s\notag\\
=&f( u_0)+\int_0^t(Y)^h(\tilde u_s)[ f]ds+\int_0^t  H_{\tau(\tilde u_s)} (\tilde u_s)[ f]\circ d \tilde W_s,
\end{align}
where we used the fact that $h_t$ has locally bounded variation, and hence 
\begin{align}
\int_0^t h^{-1}_s \circ dW_s=\int_0^t h^{-1}_s  dW_s=\tilde W_t.
\end{align}
Therefore, $\tilde u_t$ solves the SDE \req{genericSDE_no_vert} with the initial condition $u_0$, as claimed. Applying this result to the proposed limiting SDE,  \req{limiting_SDE}, completes the proof. 
\end{proof}

\section{Existence of the Zero Mass Limit}\label{convergence_proof_sec}
We are now in a position to prove convergence of the solutions  of the SDE with mass,  \reqr{SDE_u_m2}{SDE_v_m2}, to the solution of the proposed limiting SDE, \req{limiting_SDE}, as $m\rightarrow 0$.  First we need a pair of lemmas that relate metric distance on $F_O(M)$ to smooth functions.  These lemmas will allow us to prove convergence globally on the manifold $F_O(M)$, without explicitly patching together computations in local coordinates. Both lemmas make important use of  compactness of the manifold in question.

\begin{lemma}\label{f_metric_bound}
Let $(M,g)$ be a compact, connected Riemannian manifold and $d$ be the metric on $M$ induced by $g$.  For any $f\in C^\infty(M)$ there exists a $C>0$ such that
\begin{align}
|f(x)-f(y)|\leq Cd(x,y)
\end{align}
for all $x,y\in M$.
\end{lemma}
\begin{proof}
The result follows from writing
\begin{align}
f(y)-f(x)=\int_0^1 \dot{\eta}_t[f]dt=\int_0^1g_t( \dot{\eta}_t,\nabla f_t) dt
\end{align}
for any  piecewise smooth curve $\eta$ from $x$ to $y$, taking the absolute value, and using the definition of $d$ together with compactness of $M$.
\end{proof}

\begin{lemma}\label{metric_f_bound}
Let $(M,g)$ be a compact, connected Riemannian manifold with metric $d$.  For any $p>0$ there exists finitely many smooth functions $f_i$ such that
\begin{align}
d(x,y)^{p}\leq \sum_i|f_i(x)-f_i(y)|^{p}.
\end{align}
\end{lemma}
\begin{proof}
For each $x\in M$ there exists a coordinate chart $(U_x,\phi_x)$ about $x$ and a constant $C_x$ such that for $y,z\in U_x$ we have $d(y,z)\leq C_x\|\phi_x(y)-\phi_x(z)\|$ where the right hand side is the Euclidean norm (see \cite{lee2013introduction, lee2006riemannian}).  Shrinking the charts if necessary, we can assume that  $\phi_x$ extends smoothly to all of $M$.

Take another open set, $V_x$, containing $x$ with $V_x$ having compact closure in $U_x$.  By compactness of $M$ there exists a finite subcover $V_i\equiv V_{x_i}$ and there also exists $R>0$ such that $d(y,z)\leq R$ for all $y,z\in M$.  For each $i$ let $\psi_i$ be a smooth function equal to $R$ on $V_i$ with compact support in $U_i\equiv U_{x_i}$ and define $C_i\equiv C_{x_i}$, $\phi_i\equiv \phi_{x_i}$.

Take $y,z\in M$. $y\in V_i$ for some $i$.  If $z\in U_i$ then
\begin{align}
d(y,z)^2\leq C_i^2\|\phi_i(y)-\phi_i(z)\|^2=\sum_j(C_i\phi^j_i(y)-C_i\phi^j_i(z))^2.
\end{align}
Otherwise, $z\not\in U_i$, hence
\begin{align}
d(y,z)^2\leq R^2=(\psi_i(y)-\psi_i(z))^2.
\end{align}
So taking the collection of functions $C_i\phi^j_i$, $\psi_i$ gives the result for $p=2$.

For $p>0$, if we let the $N$ be the number of functions we obtained when $p=2$ then raising the $p=2$ result to the power $p/2$ gives 
\begin{align}
d(x,y)^{p}\leq (\sum_i|f_i(x)-f_i(y)|^{2})^{p/2}\leq N^{p/2}\sum_i |f_i(x)-f_i(y)|^{p}=\sum_i |g_i(x)-g_i(y)|^{p}
\end{align}
where $g_i=N^{1/2}f_i$.
\end{proof}

With these two lemmas, we can now prove the convergence result.

\begin{theorem}\label{Main_theorem}
Fix $(u_0,v_0)\in N$. Let $(u^m_t,v^m_t)$ be the unique solution to  \reqr{SDE_u_m2}{SDE_v_m2} with initial condition $(u_0,v_0)$ and let $u_t$ be the unique solution to \req{limiting_SDE} with initial condition $u_0$.  Fix $T> 0$ and a Riemannian metric tensor field on $F_O(M)$. Let $d$ be the associated metric on the connected component of $F_O(M)$ that contains $u_0$.  Then for any $q>0$ and any $0<\kappa< q/2$ we have 
\begin{align}
E[\sup_{t\in[0,T]}d(u^m_t,u_t)^{q}]=O(m^\kappa)\text{ as } m\rightarrow 0.
\end{align}

We emphasize that this result is heavily reliant on Assumption \ref{drag_assump};  the existence of a positive lower bound on the eigenvalues of the symmetric part of the damping tensor $\gamma$.
\end{theorem}
\begin{proof}
As discussed in Lemma \ref{m_existence_lemma}, a unique solution $(u^m_t,v^m_t)$ is defined for all $t\geq 0$.  Since \req{limiting_SDE} is an SDE on the compact manifold $F_O(M)$, it also has a unique solution, $u_t$, defined for all $t\geq 0$.   Both $u_t^m$ and $u_t$ are continuous, so they stay in the same connected component as $u_0$.

For $f\in C^\infty(F_O(M))$, define
\begin{align}\label{M_f_def}
M^f(u)=H_{(\gamma^{-1} F) (u)} (u)[f]+J^{\beta\alpha}(u) K^f_{\beta\alpha}(u),\hspace{2mm} Q_\eta^f(u)=H_{(\gamma^{-1}\sigma)(u)e_\eta} (u)[f]
\end{align}
where $K^f$ and $J$ were defined in \req{Kf_def} and \req{J_def} respectively. $M^f$ and $Q^f_\eta$  are smooth functions of $u$.  

\begin{remark}
We again emphasize that the indices appearing in \req{M_f_def} and in the subsequent computations represent the components in the standard basis of functions on $F_O(M)$ that are valued in the product of some $\mathbb{R}^l$'s.  Though we employ the summation convention, these expressions do not represent (contractions of) tensors on $M$ or $F_O(M)$.  In particular, the following convergence proof is global in nature and does not directly utilize any computations in local coordinates on the base manifold or its frame bundle.
\end{remark}

Using \req{dfu_m} and \req{limiting_SDE_Ito} and rearranging we obtain
\begin{align}
f(u^m_t)-f(u_t)=&\int_0^t\left(M^f(u^m_s)-M^f(u_s)\right)ds\\
&+\int_0^t\left(Q_\eta^f(u^m_s)-Q_\eta^f(u_s) \right)dW^\eta_s+R^f(t),\notag
\end{align}
where we have grouped the momentum-dependent terms in the quantity
\begin{align}
R^f(t)=&- \int_0^tK^f_{\beta\alpha}(u^m_s)G^{\beta\alpha}_{\mu\nu}(u^m_s)d((p^m_s)^\nu (p^m_s)^\mu)\\
&+\int_0^tK^f_{\beta\alpha}(u^m_s)G^{\beta\alpha}_{\mu\nu}(u^m_s)\left((p^m_s)^\nu F^\mu(u^m_s)+(p^m_s)^\mu F^\nu(u^m_s)\right)ds\notag\\
&+\int_0^tK^f_{\beta\alpha}(u^m_s)G^{\beta\alpha}_{\mu\nu}(u^m_s)\left((p^m_s)^\mu\sigma^\nu_\eta(u^m_s) +(p^m_s)^\nu\sigma^\mu_\eta(u^m_s) \right)dW^\eta_s\notag\\
&-H_{\gamma^{-1}(u^m_t)p^m_t} (u^m_t)[f]+H_{\gamma^{-1}(u_0)p_0} (u_0)[f].\notag
\end{align}
For now, restrict to $q=2p$ with $p> 1$. For any $t\leq T$,
\begin{align}
&E[\sup_{0\leq s\leq t}|f(u^m_s)-f(u_s)|^{2p}]\leq3^{2p-1}\bigg(E[\left(\int_0^t|M^f(u^m_s)-M^f(u_s)|ds\right)^{2p}]\notag\\
&+E[\sup_{0\leq s\leq t}|\int_0^s Q_\eta^f(u^m_r)-Q_\eta^f(u_r) dW^\eta_r|^{2p}]+E[\sup_{0\leq s\leq t}|R^f(s)|^{2p}]\bigg).
\end{align}
Applying the H\"older's inequality with  exponents $2p$ and $2p/(2p-1)$ to the first  term and the Burkholder-Davis-Gundy inequality to the second we get
\begin{align}
&E[\sup_{0\leq s\leq t}|f(u^m_s)-f(u_s)|^{2p}]\leq3^{2p-1}\bigg(T^{2p-1}E[\int_0^t|M^f(u^m_s)-M^f(u_s)|^{2p}ds]\notag\\
&+CE[\left(\int_0^t\sum_\eta |Q_\eta^f(u^m_s)-Q_\eta^f(u_s)|^2 ds\right)^{p}]+E[\sup_{0\leq s\leq t}|R^f(s)|^{2p}]\bigg).
\end{align}
We have assumed $p>1$, so we can use H\"older's inequality with indices $p$ and $p/(p-1)$ on the second term.
\begin{align}
E[\left(\int_0^t\sum_\eta |Q_\eta^f(u^m_s)-Q_\eta^f(u_s)|^2 ds\right)^{p}]\leq T^{p-1}E[\int_0^t\left(\sum_\eta |Q_\eta^f(u^m_s)-Q_\eta^f(u_s)|^{2}\right)^p ds].
\end{align}
Therefore
\begin{align}\label{f_estimate}
&E[\sup_{0\leq s\leq t}|f(u^m_s)-f(u_s)|^{2p}]\leq3^{2p-1}\bigg(T^{2p-1}E[\int_0^t|M^f(u^m_s)-M^f(u_s)|^{2p}ds]\\
&+C(kT)^{p-1}\sum_\eta E[\int_0^t |Q_\eta^f(u^m_s)-Q_\eta^f(u_s)|^{2p} ds]+E[\sup_{0\leq s\leq t}|R^f(s)|^{2p}]\bigg)\notag
\end{align}
where $k$ is the number of terms in the sum over $\eta$.

The integrands of the terms involving $M^f$ and $Q^f_\eta$ are of the form $|g(u^m_s)-g(u_s)|^{2p}$ for some smooth functions $g$, depending on $f$. Therefore, for a fixed $f$, by Lemma \ref{f_metric_bound} there exists $C_f>0$ such that
\begin{align} \label{f_estimate_final}
&E[\sup_{0\leq s\leq t}|f(u^m_s)-f(u_s)|^{2p}]\\
\leq&3^{2p-1}(T^{2p-1}+Ck^{p}T^{p-1} )C_f\int_0^tE[d(u^m_s,u_s)^{2p}]ds+3^{2p-1}E[\sup_{0\leq s\leq t}|R^f(s)|^{2p}].\notag
\end{align}

$F_O(M)$ is compact and we have equipped it with a Riemannian metric tensor.  Therefore, by Lemma \ref{metric_f_bound}, there exists finitely many smooth functions, $f_i$, such that for any $u_1,u_2$ in the the connected component of $F_O(M)$ containing $u_0$ we have
\begin{align}
d(u_1,u_2)^{2p}\leq \sum_i|f_i(u_1)-f_i(u_2)|^{2p}.
\end{align}
  Hence, applying \req{f_estimate_final} to each of the functions $f_i$, we obtain
\begin{align}
&E[\sup_{0\leq s\leq t}d(u^m_s,u_s)^{2p}]\leq\sum_i E[\sup_{0\leq s\leq t}|f_i(u^m_s)-f_i(u_s)|^{2p}]\\
\leq &\sum_i C_i \int_0^tE[d(u^m_s,u_s)^{2p}]ds+3^{2p-1}\sum_i E[\sup_{0\leq s\leq t}|R^{f_i}(s)|^{2p}]\notag
\end{align}
for some constants $C_i$.
Therefore there exists a $C>0$ such that
\begin{align}
&E[\sup_{0\leq s\leq t}d(u^m_s,u_s)^{2p}]\\
\leq& C\int_0^tE[\sup_{0\leq r\leq s}d(u^m_r,u_r)^{2p}]ds+3^{2p-1}\sum_iE[\sup_{0\leq s\leq T}|R^{f_i}(s)|^{2p}]\notag
\end{align}
for all $0\leq t\leq T$. We can apply Gronwall's inequality to obtain
\begin{align}\label{Gronwall_result}
&E[\sup_{0\leq s\leq T}d(u^m_s,u_s)^{2p}]\leq 3^{2p-1}\sum_iE[\sup_{0\leq s\leq T}|R^{f_i}(s)|^{2p}]e^{CT}.
\end{align}

Fix $0<\kappa< p$. The sum in \req{Gronwall_result} has finitely many terms, so, if we can prove that  
\begin{align}
E[\sup_{0\leq s\leq T}|R^f(s)|^{2p}]=O(m^\kappa) \text{ as } m\rightarrow 0
\end{align}
for every smooth $f$, then the claim will follow for $q=2p$, $p>1$.

Fix  $f\in C^\infty(F_O(M))$ and compute 
\begin{align}
&E[\sup_{0\leq t\leq T}|R^f(t)|^{2p}]\\
\leq&5^{2p-1}\bigg(E[\sup_{0\leq t\leq T}|\int_0^tK^f_{\beta\alpha}(u^m_s)G^{\beta\alpha}_{\mu\nu}(u^m_s)d((p^m_s)^\nu (p^m_s)^\mu)|^{2p}]\notag\\
&+E[\left(\int_0^T|K^f_{\beta\alpha}(u^m_s)G^{\beta\alpha}_{\mu\nu}(u^m_s)\left((p^m_s)^\nu F^\mu(u^m_s)+(p^m_s)^\mu F^\nu(u^m_s)\right)|ds\right)^{2p}]\notag\\
&+E[\sup_{0\leq t\leq T}|\int_0^tK^f_{\beta\alpha}(u^m_s)G^{\beta\alpha}_{\mu\nu}(u^m_s)\left((p^m_s)^\nu\sigma^\mu_\eta(u^m_s) +(p^m_s)^\mu\sigma^\nu_\eta(u^m_s) \right)dW^\eta_s|^{2p}]\notag\\
&+E[\sup_{0\leq t\leq T}|H_{\gamma^{-1}(u^m_t)p^m_t} (u^m_t)[f]|^{2p}]+E[|H_{\gamma^{-1}(u_0)p_0} (u_0)[f]|^{2p}]\bigg).\notag
\end{align}
We will consider each term individually. Proposition \ref{quad_p_int_bound} implies that
\begin{align}
E[\sup_{0\leq t\leq T}|\int_0^tK^f_{\beta\alpha}(u^m_s)G^{\beta\alpha}_{\mu\nu}(u^m_s)d((p^m_s)^\nu (p^m_s)^\mu)|^{2p}] =O(m^{p}).
\end{align}
Using H\"older's inequality, boundedness of continuous functions on the compact manifold $F_O(M)$, and Proposition \ref{sup_E_p_bound}, the second term  can be bounded as follows,
\begin{align}
&E[\left(\int_0^T|K^f_{\beta\alpha}(u^m_s)G^{\beta\alpha}_{\mu\nu}(u^m_s)\left((p^m_s)^\nu F^\mu(u^m_s)+(p^m_s)^\mu F^\nu(u^m_s)\right)|ds\right)^{2p}]\notag\\
\leq&T^{2p-1}\!\!\!\int_0^T\!\!\!E[|\left((K^f_{\beta\alpha}G^{\beta\alpha}_{\nu\mu} F^\nu)(u^m_s)+(K^f_{\beta\alpha}G^{\beta\alpha}_{\mu\nu}F^\nu)(u^m_s)\right)(p^m_s)^\mu|^{2p}]ds\notag\\
\leq&CT^{2p} \sup_{0\leq s\leq T}E[\|p^m_s\|^{2p}]=O(m^p).
\end{align}
The third term can be bounded using the  Burkholder-Davis-Gundy inequality, similarly to the way the second term in \req{f_estimate} was estimated, together with  Proposition \ref{sup_E_p_bound}.
\begin{align}
&E[\sup_{0\leq t\leq T}|\int_0^tK^f_{\beta\alpha}(u^m_s)G^{\beta\alpha}_{\mu\nu}(u^m_s)\left((p^m_s)^\nu\sigma^\mu_\eta(u^m_s) +(p^m_s)^\mu\sigma^\nu_\eta(u^m_s) \right)dW^\eta_s|^{2p}]\notag\\
\leq &C(kT)^{p-1}\sum_\eta E[\int_0^T|\left(K^f_{\beta\alpha}(u^m_s)G^{\beta\alpha}_{\nu\mu}(u^m_s)\sigma^\nu_\eta(u^m_s)\right.\\
&\left. +K^f_{\beta\alpha}(u^m_s)G^{\beta\alpha}_{\mu\nu}(u^m_s)\sigma^\nu_\eta(u^m_s)\right)(p^m_s)^\mu |^{2p}ds]\notag\\
\leq  &\tilde{C}T^{p-1}\int_0^TE[\|p^m_s\|^{2p}]ds\leq \tilde{C}T^{p}\sup_{0\leq s\leq T}E[\|p^m_s\|^{2p}]=O(m^p).
\end{align}
Compactness of $F_O(M)$ and Proposition \ref{E_sup_p_bound} imply a bound on the fourth term:
\begin{align}
&E[\sup_{0\leq t\leq T}|H_{\gamma^{-1}(u^m_t)p^m_t} (u^m_t)[f]|^{2p}]= E[\sup_{0\leq t\leq T}|H_{\gamma^{-1}(u^m_t)e_\nu } (u^m_t)[f](p^m_t)^\nu|^{2p}]\notag\\
\leq&C E[\sup_{0\leq t\leq T}\|p^m_t\|^{2p}]=O(m^\kappa).
\end{align}
Finally, the last term is
\begin{align}
E[|H_{\gamma^{-1}(u_0)p_0} (u_0)[f]|^{2p}]=m^{2p}E[|H_{\gamma^{-1}(u_0)v_0} (u_0)[f]|^{2p}]=O(m^{2p}).
\end{align}
This completes the proof for $q=2p$, $p>1$. Similar to Proposition \ref{sup_E_p_bound}, an application of H\"older's inequality  gives the result for all $q>0$.
\end{proof}
As a corollary, we get convergence in probability and in law on compact time intervals.
\begin{corollary}
For any $T>0$, the mass-dependent process restricted to the compact time interval $[0,T]$,  $u^m|_{[0,T]}$, converges in probability (and hence also in law)  to $u|_{[0,T]}$ as $m\rightarrow 0$.

 Convergence in probability on the path space $C([0,T],F^i_O(M))$, where $F^i_O(M)$ is a connected component of $F_O(M)$, is defined through the metric
\begin{align}
d_T(u,\tau)=\sup_{t\in[0,T]} d(u(t),\tau(t))
\end{align}
 where $d$ is the metric on $F^i_O(M)$ induced by a choice of Riemannian metric tensor on $F_O(M)$.  
\end{corollary}

\subsection*{Acknowledgments}

The authors are grateful to K. Gaw\c edzki, P. Grzegorczyk, X.-M. Li, A. McDaniel for stimulating discussions. We thank an editor of AHP for pointing out several references.  The research of JW was supported in part by US NSF Grant DMS-131271.  A part of his work on this article was supported by US NSF Grant DMS-1440140 while he was in residence at the Mathematical Sciences Research Institute in Berkeley during the Fall 2015 semester.

\bibliographystyle{ieeetr}
\bibliography{refs}

\end{document}